\documentclass[a4paper,11pt,DIV12]{scrartcl}

\usepackage[paper=a4paper,left=2.5cm,right=2.5cm,top=3.0cm,bottom=3.0cm,footskip=1.0cm]{geometry}
\usepackage{kMeans}

\ifpdf
  \RequirePackage[pdftex,linktocpage,
    pdfauthor={David Arthur, Bodo Manthey, Heiko Roeglin},
    pdftitle={k-Means has Polynomial Smoothed Complexity},
    breaklinks]{hyperref}
\else
  \RequirePackage[breaklinks]{hyperref}
\fi

\title{$k$-Means has Polynomial Smoothed Complexity}

\author{David Arthur \\[\addressminusB]
        \footnotesize Stanford University \\[\addressminus] \footnotesize
        Dept.~of Computer Science \\[\addressminus] \footnotesize
        \texttt{darthur@cs.stanford.edu}
   \and
        Bodo Manthey \\[\addressminusB]
        \footnotesize University of Twente \\[\addressminus] \footnotesize
        Dept.~of Applied Mathematics \\[\addressminus] \footnotesize
        \texttt{b.manthey@utwente.nl}
   \and
        Heiko R\"oglin\thanks{Supported by a fellowship within the
        Postdoc-Program of the German Academic Exchange Service (DAAD).} \\[\addressminusB]
        \footnotesize Maastricht University \\[\addressminus] \footnotesize
        Dept.~of Quantitative Economics\\[\addressminus] \footnotesize
        \texttt{heiko@roeglin.org}}

\begin{document}

\maketitle

\begin{abstract}
The $k$-means method is one of the most widely used clustering algorithms, drawing
its popularity from its speed in practice. Recently, however, it was
shown to have exponential worst-case running time. In order to close the gap
between practical performance and theoretical analysis, the $k$-means method has
been studied in the model of smoothed analysis. But even the smoothed
analyses so far are unsatisfactory as the bounds are still
super-polynomial in the number $n$ of data points.

In this paper, we settle the smoothed running time of the $k$-means method. We show
that the smoothed number of iterations is bounded by a polynomial in $n$ and
$1/\sigma$, where $\sigma$ is the standard deviation of the Gaussian
perturbations. This means that if an arbitrary input data set is randomly
perturbed, then the $k$-means method will run in expected polynomial time on that
input set.
\end{abstract}

\section{Introduction}

Clustering is a fundamental problem in computer science with applications ranging
from biology to information retrieval and data compression. In a clustering problem,
a set of objects, usually represented as points in a high-dimensional space $\RR^d$, is to be
partitioned such that objects in the same group share similar properties.
The $k$-means method is a traditional clustering algorithm,
which is based on ideas by Lloyd~\cite{Lloyd}.
It begins with an arbitrary clustering based
on $k$ centers in $\RR^d$, and then repeatedly makes local
improvements until the clustering stabilizes. The algorithm is
greedy and as such, it offers virtually no accuracy guarantees.
However, it is both very simple and very fast, which makes it
appealing in practice.
Indeed, one recent survey of data mining techniques states that the $k$-means
method ``is by far the most popular clustering algorithm used in scientific and
industrial applications''~\cite{Berkhin}.

However, theoretical analysis has long been at stark contrast
with what is observed in
practice. In particular, it was recently shown that the
worst-case running time of the $k$-means method is $2^{\Omega(n)}$ even on
two-dimensional instances~\cite{Vattani}. Conversely, the only upper bounds known for the
general case are $k^n$ and $n^{O(kd)}$. Both upper bounds
are based entirely on the
trivial fact that the $k$-means method never encounters the same clustering
twice~\cite{InabaEA:WeightedVoronoi:2000}. In contrast, Duda
et al.\ state that
the number of iterations until the clustering stabilizes
is often linear or even sublinear in $n$ on practical data sets~\cite[Section 10.4.3]{Duda}.
The only known polynomial upper bound, however,
applies only in one dimension and only for certain inputs \cite{HarPeled}.

So what does one do when worst-case analysis is at odds with what is
observed in practice? We turn to the smoothed analysis of Spielman
and Teng~\cite{SpielmanTeng:SmoothedAnalysisWhy:2004}, which considers
the running time after first
randomly perturbing the input. Intuitively, this models how fragile
worst-case instances are and if they could reasonably arise in
practice.
In addition to the original work on the simplex algorithm,
smoothed analysis has been applied successfully in other contexts,
e.g., for the ICP algorithm~\cite{ArthurVassilvitskii:ICP:2009},
online algorithms~\cite{BecchettiEA}, the knapsack problem~\cite{BeierV04},
and the 2-opt heuristic for the TSP~\cite{EnglertRV07}.

The $k$-means method is in fact a perfect candidate for smoothed analysis: it is
extremely widely used, it runs very fast in practice, and yet the worst-case
running time is exponential. Performing this analysis has proven very
challenging however. It has been initiated by Arthur and Vassilvitskii who
showed that the smoothed running time of the $k$-means method is polynomially
bounded in $n^{k}$ and $1/\sigma$, where $\sigma$ is the standard deviation of
the Gaussian perturbations~\cite{ArthurVassilvitskii:ICP:2009}. The term
$n^{k}$ has been improved to $\min(n^{\sqrt k}, k^{kd} \cdot n)$ by Manthey and
R\"oglin~\cite{MantheyRoeglin:kMeans:2009}. Unfortunately, this bound remains
exponential even for relatively small values of $k$.
In this paper we settle the smoothed running time of the
$k$-means method: We prove that it is polynomial in $n$ and $1/\sigma$.
The exponents in the polynomial are unfortunately too large
to match the practical observations,
but this is in line with other
works in smoothed analysis, including Spielman and Teng's original
analysis of the simplex method \cite{SpielmanTeng:SmoothedAnalysisWhy:2004}.
The arguments presented here, which reduce the smoothed
upper bound from exponential to polynomial, are intricate enough
without trying to optimize constants, even in the exponent. However, we hope and
believe that our work can be used as a basis for proving tighter results
in the future.

\subsection{\texorpdfstring{$k$}{k}-Means Method}

An input for the $k$-means method is a set $\points \subseteq \RR^d$
of $n$ data points. The algorithm outputs
$k$ centers $c_1,\ldots,c_k\in\RR^d$ and
a partition of $\points$ into $k$ clusters $\cluster 1,\ldots,\cluster k$.
The $k$-means method proceeds as follows:
\begin{enumerate}
\item Select cluster centers $c_1, \ldots, c_k\in\RR^d$ arbitrarily.
\item \label{next} Assign every $x \in \points$ to the cluster
$\cluster i$ whose cluster center $c_i$ is closest to it, i.e., $\|x-c_i\|
\le \|x-c_j\|$ for all $j \neq i$.
\item Set $c_i = \frac{1}{|\cluster i|}\sum_{x \in \cluster i} x$.
\item If clusters or centers have changed, goto~\ref{next}. Otherwise,
      terminate.
\end{enumerate}
In the following, an \emph{iteration} of $k$-means refers to
one execution of step~2 followed by step~3.
A slight technical subtlety in the implementation of the algorithm
is the possible event that a cluster loses all its points in Step~2.
There exist some strategies to deal with this
case~\cite{HarPeled}.
For simplicity, we use the strategy of removing clusters that serve no points
and continuing with the remaining clusters.

If we define $c(x)$ to be the center closest to a data point $x$,
then one can check that each step of the algorithm decreases the
following potential function:
\[
    \Psi = \sum_{x \in \mathcal{X}} \norm{x - c(x)}^2 \: .
\]
The essential observation for this is the following:
If we already have cluster centers $c_1, \ldots, c_k\in\RR^d$
representing clusters,
then every data point should be assigned to the cluster whose
center is nearest to it to minimize $\Psi$. On the other hand, given
clusters $\cluster 1, \ldots, \cluster k$, the centers
$c_1, \ldots, c_k$ should be chosen as the centers of
mass of their respective clusters in order to minimize the potential.

In the following, we will speak of $k$-means rather than of the $k$-means
method for short.
The worst-case running time of $k$-means is bounded from above
by $(k^2n)^{kd} \leq n^{3kd}$, which follows from Inaba et
al.~\cite{InabaEA:WeightedVoronoi:2000} and Warren~\cite{Warren}.
(The bound of $O(n^{kd})$ frequently stated
in the literature holds
only for constant values for $k$ and $d$, but in this paper
$k$ and $d$ are allowed to grow.)
This upper bound is based solely on the observation
that no clustering occurs twice during an execution of $k$-means
since the potential decreases in every iteration.
On the other hand, the worst-case number of iterations has been
proved to be $\exp(\sqrt n)$
for $d \in \Omega(\sqrt n)$~\cite{ArthurVassilvitskii:HowSlow:2006}.
This has been improved recently to $\exp(n)$ for
$d \geq 2$~\cite{Vattani}.

\subsection{Related Work}

The problem of finding good $k$-means clusterings allows
for polynomial-time approximation schemes~\cite{Badoiu,Matousek,Kumar}
with various dependencies of the running time on $n$, $k$, $d$, and
the approximation ratio $1+\eps$. The running times of these
approximation schemes depend exponentially on $k$.
Recent research on this subject
also includes the work by Gaddam et al.~\cite{kmeansExample1} and
Wagstaff et al.~\cite{kmeansExample3}.
However, the most widely used algorithm for $k$-means clustering is still
the $k$-means method due to its simplicity and speed.

Despite its simplicity, the $k$-means method itself and variants thereof are
still the subject of research~\cite{kmeansExample2,kMeansPP,Ostrovsky}.
Let us mention in particular the work by Har-Peled and Sadri~\cite{HarPeled}
who have shown that a certain variant of the $k$-means method runs in
polynomial time on certain instances. In their variant, a data point is said to be
$(1+\eps)$-misclassified if the distance to its current cluster center is
larger by a factor of more than $(1+\eps)$ than the distance to its closest
center. Their \emph{lazy $k$-means method} only reassigns points that are
$(1+\eps)$-misclassified. In particular, for $\eps=0$, lazy $k$-means and
$k$-means coincide. They show that the number of steps of the
lazy $k$-means method is polynomially bounded in the number of data points,
$1/\eps$, and the spread of the point set (the spread of a point set is the
ratio between its diameter and the distance between its closest pair).

In an attempt to reconcile theory and practice, Arthur and
Vassilvitskii~\cite{ArthurVassilvitskii:ICP:2009} performed the first
smoothed analysis of the $k$-means method: If the data points
are perturbed by Gaussian perturbations of standard
deviation $\sigma$, then the smoothed number of iterations is
polynomial in $n^k$, $d$, the diameter of the point set, and $1/\sigma$.
However, this bound is still super-polynomial in the number $n$ of data points.
They conjectured that $k$-means has indeed polynomial smoothed
running time, i.e., that the smoothed number of iterations is bounded by some
polynomial in $n$ and $1/\sigma$.

Since then, there has been only partial success in proving the conjecture.
Manthey and R\"oglin improved the smoothed running time
bound by devising two bounds~\cite{MantheyRoeglin:kMeans:2009}:
The first is polynomial in $n^{\sqrt k}$ and $1/\sigma$. The second is
$k^{kd} \poly(n, 1/\sigma)$, where the
degree of the polynomial is independent of $k$ and~$d$.
Additionally, they proved a polynomial bound for the smoothed running time
of $k$-means on one-dimensional instances.

\subsection{Our Contribution}

We prove that the $k$-means method has polynomial smoothed running time.
This finally proves Arthur and Vassilvitskii's
conjecture~\cite{ArthurVassilvitskii:ICP:2009}.

\begin{theorem}
\label{thm:main}
Fix an arbitrary set $\points'\subseteq[0,1]^d$ of $n$ points
and assume that each point in $\points'$ is
independently perturbed by a normal distribution with
mean $0$ and standard deviation $\sigma$,
yielding a new set $\points$ of points.
Then the expected running time of $k$-means on $\points$
is bounded by a polynomial in $n$ and $1/\sigma$.
\end{theorem}

We did not optimize the exponents in the polynomial as the arguments
presented here, which reduce the smoothed upper bound from exponential
to polynomial, are already intricate enough and would not yield
exponents matching the experimental observations even when optimized.
We hope that similar to the smoothed analysis of the simplex algorithm,
where the first polynomial bound~\cite{SpielmanTeng:SmoothedAnalysisWhy:2004}
stimulated further research culminating in Vershynin's
improved bound~\cite{Vershynin09}, our result here will also be the first step
towards a small polynomial bound for the smoothed running time of $k$-means.
As a reference, let us mention that the upper bound on the
expected number of iterations following from our proof is
\[
 O\left(\frac{n^{34} \log^4(n) k^{34}d^{8}}{\sigma^6}\right)\:.
\]

The idea is to prove, first, that the potential after one
iteration is bounded by some polynomial and, second,
that the potential decreases by some polynomial amount in
every iteration (or, more precisely, in every sequence of a few
consecutive iterations). To do this, we prove upper bounds on the
probability that the minimal improvement is small. The main challenge
is the huge number of up to $n^{3kd}$ possible clusterings.
Each of these clusterings yields a potential iteration of $k$-means,
and a simple union bound over all of them is too weak to yield
a polynomial bound.

To prove the bound of
$\poly(n^{\sqrt k}, 1/\sigma)$~\cite{MantheyRoeglin:kMeans:2009},
a union bound was taken over the $n^{3kd}$ clusterings.
This is already
a technical challenge as the set of possible clusterings
is fixed only after the points are fixed.
To show a polynomial bound, we reduce the number of
cases in the union bound by introducing the notion of
\emph{transition blueprints}. Basically, every iteration
of $k$-means can be described by a transition
blueprint. The blueprint describes the iteration only roughly,
so that several iterations are described
by the same blueprint. Intuitively, iterations with the same
transition blueprint are correlated in the sense that
either all of them make a small improvement or none of them do.
This dramatically reduces the number of cases that
have to be considered in the union bound. On the other hand,
the description conveyed by a blueprint is still precise enough
to allow us to bound the probability that any iteration described by it
makes a small improvement.

We distinguish between several types of iterations, based on
which clusters exchange how many points.
Sections~\ref{sec:smallbalance} to~\ref{ssec:degenerateBlueprints} deal
with some special cases of iterations that
need separate analyses.

After that, we analyze the general case
(Section~\ref{ssec:other}). The difficulty in this analysis is
to show that every transition blueprint contains
``enough randomness''. We need to show that this randomness
allows for sufficiently tight upper bounds on the probability
that the improvement obtained from any iteration
corresponding to the blueprint is small.

Finally, we put the six sections together to prove that
$k$-means has polynomial smoothed running time (Section~\ref{ssec:merging}).

\section{Preliminaries}
\label{sec:prelim}

For a finite set
$X \subseteq \RR^d$, let $\mass(X) = \frac{1}{|X|} \sum_{x\in X} x$ be the
center of mass of the set $X$.
If $H \subseteq \RR^d$ is a hyperplane and $x \in \RR^d$ is a single
point, then $\dist(x, H) = \min \{\|x-y\| \mid y \in H\}$ denotes the
distance of the point $x$ to the hyperplane $H$.

For our smoothed analysis, an adversary specifies an instance
$\points' \subseteq [0,1]^d$ of $n$ points.
Then each point $x' \in \points'$ is perturbed by adding an independent
$d$-dimensional Gaussian random vector with standard deviation $\sigma$
to $x'$ to obtain the data point $x$. These perturbed points
form the input set $\points$.
For convenience we assume that $\sigma \leq 1$.
This assumption is without loss of generality as for larger values
of $\sigma$, the smoothed running time can only be smaller than for
$\sigma=1$~\cite[Section 7]{MantheyRoeglin:kMeans:2009}.
Additionally we assume $k\le n$ and $d\le n$:
First, $k \leq n$ is satisfied after the first iteration since at most
$n$ clusters can contain any points.
Second, $k$-means is known to have polynomial smoothed complexity
for $d \in \Omega(n/\log n)$~\cite{ArthurVassilvitskii:HowSlow:2006}.
The restriction of the adversarial points to be in $[0,1]^d$ is necessary as,
otherwise, the adversary can diminish the effect of the perturbation by placing
all points far apart from each other. Another way to cope with this problem is to
state the bounds in terms of the diameter of the adversarial
instance~\cite{ArthurVassilvitskii:ICP:2009}. However,
to avoid having another parameter, we have chosen the former model.

Throughout the following, we assume that the perturbed point set
$\points$ is contained in some hypercube of side-length
$D$, i.e., $\points \subseteq [-D/2, D/2]^d = \cube$.
We choose $D$ such that the probability of $\points\not\subseteq\cube$
is bounded from above by $n^{-3kd}$. Then, as the worst-case number
of iterations is bounded by $n^{3kd}$~\cite{InabaEA:WeightedVoronoi:2000},
the event $\points\not\subseteq\cube$ contributes only an insignificant
additive term of $+1$ to the expected number of iterations, which we
ignore in the following.

Since Gaussian random vectors are heavily concentrated around their
mean and all means are in $[0,1]^d$, we can choose $D=\sqrt{90kd\ln(n)}$
to obtain the desired failure probability for $\points\not\subseteq\cube$,
as shown by the following calculation, in which $Z$ denotes a
one-dimensional Gaussian random variable with mean 0 and standard deviation 1:
\begin{align*}
  \Pr{\points\not\subseteq\cube} & \le nd\cdot\Pr{|Z|\ge D/2-1}
  \le 2nd\cdot\Pr{Z\ge D/3}\\
  & \le \frac{2nd}{\sqrt{2\pi}}\cdot\exp(-D^2/18)
  \le n^2\cdot\exp(-D^2/18) \le n^{-3kd}\:,
\end{align*}
where we used $k\ge2$, $d\ge2$, and the tail
bound $\Prob[Z\ge z] \le \frac{\exp(-z^2/2)}{z\sqrt{2\pi}}$
for Gaussians~\cite{Durrett}.

For our smoothed analysis, we use essentially three properties of Gaussian
random variables. Let $X$ be a $d$-dimensional Gaussian random variable
with standard deviation $\sigma$.
First, the probability that $X$ assumes a value
in any fixed ball of radius $\eps$ is at most $(\eps/\sigma)^d$.
Second, let $b_1, \ldots, b_{d'} \in \RR^d$ be orthonormal vectors
for some $d' \leq d$. Then the vector $(b_1\cdot X, \ldots, b_{d'}\cdot X)\in\RR^{d'}$
is a $d'$-dimensional Gaussian random variable with the same
standard deviation $\sigma$.
Third, let $H$ be any hyperplane. Then the probability that a
Gaussian random variable assumes a value that is within a distance
of at most $\eps$ from $H$ is bounded by $\eps/\sigma$.
This follows also from the first two properties if we choose $d'=1$
and $b_1$ to be the normal vector of $H$.

We will often upper-bound various probabilities, and it will be
convenient to reduce the exponents in these bounds.
Under certain conditions, this can be done safely regardless of whether the
base is smaller or larger than 1.

\begin{fact} \label{fact:kmubProbExponents}
    Let $p$ be a probability, and let $A, c, b, e,$ and $e'$ be positive real numbers satisfying $c \ge 1$ and $e \ge e'$. If
    $p \le A + c \cdot b^e$, then it is also true that $p \le A + c \cdot b^{e'}$.
\end{fact}

\begin{proof}
    If $b$ is at least 1, then $A + c \cdot b^{e'}\ge 1$ and it is trivially true that $p \le A + c \cdot b^{e'}$. Otherwise, $b^e \leq b^{e'}$, and the
    result follows.
\end{proof}

\subsection{Potential Drop in an Iteration of \texorpdfstring{$k$}{k}-Means}

During an iteration of the $k$-means method there are
two possible events that can lead to a significant potential drop:
either one cluster center moves significantly, or a data point
is reassigned from one cluster to another and this point has a significant
distance from the bisector of the clusters (the bisector is the hyperplane
that bisects the two cluster centers). In the following
we quantify the potential drops caused by these events.

The potential drop caused by reassigning a data point $x$ from one cluster to
another can be expressed in terms of the distance of $x$ from the bisector of the
two cluster centers and the distance of these two centers. The following lemma
follows from basic linear algebra
(cf., e.g., \cite[Proof of Lemma 4.5]{MantheyRoeglin:kMeans:2009}).

\begin{lemma}
\label{lem:epsdeltadrop}
Assume that, in an iteration of $k$-means,
a point $x \in \points$ switches from $\cluster i$
to $\cluster j$. Let $c_i$ and $c_j$ be the centers of these clusters,
and let $H$ be their bisector.
Then reassigning $x$ decreases the potential by
$2\cdot\|c_i-c_j\|\cdot\dist(x, H)$.
\end{lemma}

The following lemma, which also follows from basic linear algebra,
reveals how moving a cluster center to the center of mass decreases the potential.

\begin{lemma}[Kanungo et al.~\cite{KMNPSW}]
\label{lem:movement}
Assume that the center of a cluster $\calC$ moves from $c$
to $\mass(C)$ during an iteration of $k$-means, and
let $|\calC|$ denote the number of points in $\calC$ when the movement occurs.
Then the potential decreases by $|\calC|\cdot \|c - \mass(C)\|^2$.
\end{lemma}

\subsection{The Distance between Centers}

As the distance between two cluster centers plays an important role
in Lemma~\ref{lem:epsdeltadrop}, we analyze how close together
two simultaneous centers can be during the execution of $k$-means.
This has already been analyzed
implicitly~\cite[Proof of Lemma 3.2]{MantheyRoeglin:kMeans:2009},
but the variant below gives stronger bounds.
From now on, when we refer to a $k$-means
iteration, we will always mean an iteration \emph{after the first one}. By
restricting ourselves to this case, we ensure that the centers at the beginning
of the iteration are the centers of mass of actual clusters, as opposed to
the arbitrary choices that were used to seed $k$-means.
\begin{definition}
    Let $\delta_\eps$ denote the minimum distance between two cluster centers at the beginning of a $k$-means iteration in which (1) the
    potential $\pot$ drops by at most $\eps$, and (2) at least one data point switches between the clusters corresponding to these
    centers.
\end{definition}

\begin{lemma} \label{lemma:kmubCloseCenters}
    Fix real numbers $Y \ge 1$ and $e \ge 2$. Then, for any $\eps \in [0,1]$,
    \[
        \Pr{\delta_\eps \le Y\eps^{1/e}} \le \eps \cdot \left(\frac{O(1) \cdot n^5Y}{\sigma}\right)^e.
    \]
\end{lemma}

\begin{proof}
    Consider a $k$-means iteration $I$ that results in a potential drop of at most $\eps$, and let $I_0$ denote the previous iteration. Also consider a
    fixed pair of clusters that exchange at least one data point during $I$. We define the following:
    \begin{itemize}
        \item Let $a_0$ and $b_0$ denote the centers of the two clusters at the beginning of iteration $I_0$ and let $H_0$ denote the hyperplane
        bisecting $a_0$ and $b_0$.
        \item Let $A$ and $B$ denote the set of data points in the two clusters at the beginning of iteration $I$. Note that $H_0$ splits $A$ and $B$.
        \item Let $a$ and $b$ denote the centers of the two clusters at the beginning of iteration $I$, and let $H$ denote the hyperplane bisecting
        $a$ and $b$. Note that $a = \mass(A)$ and $b = \mass(B)$.
        \item Let $A'$ and $B'$ denote the set of data points in the two clusters at the end of iteration $I$. Note that $H$ splits $A'$ and $B'$.
        \item Let $a'$ and $b'$ denote the centers of the two clusters at the end of iteration $I$. Note that $a' = \mass(A')$ and $b' = \mass(B')$.
    \end{itemize}
    Also let $t = 3d + \floor{e}$. Now suppose we have $\norm{a - b} \le Y\eps^{1/e}$.

    First we consider the case $|A' \cup A| \ge t+1$. We claim that every point in $A$ must be within
    distance $nY\eps^{1/e}$ of $H_0$. Indeed, if this were not true, then since $H_0$ splits $A$ and $B$, and since $a = \mass(A)$ and $b = \mass(B)$,
    we would have $\norm{a - b} \ge \dist(a, H_0) > \frac{nY\eps^{1/e}}{|A|} \ge Y\eps^{1/e}$, giving a contradiction. Furthermore, as $I$ results in a
    potential drop of at most $\eps$,
    Lemma~\ref{lem:movement} implies that $\norm{a' - a}, \norm{b' - b} \le \sqrt{\eps}$, and therefore,
    \[
        \norm{a' - b'} \le \norm{a' - a} + \norm{a - b} + \norm{b - b'} \le Y\eps^{1/e} + 2\sqrt{\eps} \le 3Y\eps^{1/e}.
    \]
    In particular, we can repeat the above argument to see that every point in
    $A'$ must be within distance $3nY\eps^{1/e}$ of $H$. This means that there are two
    hyperplanes such that every point in $A \cup A'$ is within distance
    $3nY\eps^{1/e}$ of at least one of these hyperplanes.
    Following the arguments by Arthur and
    Vassilvitskii~\cite[Proposition~5.9]{ArthurVassilvitskii:ICP:2009},
    we obtain that the probability that there
    exists a set $A \cup A'$ of size $t+1$ with this property is at most
    \begin{eqnarray}
        n^{t+1} \cdot \left(\frac{12dnY\eps^{1/e}}{\sigma}\right)^{t+1-2d}
        &=& n^{3d+\floor{e}+1} \cdot \left(\frac{12dnY\eps^{1/e}}{\sigma}\right)^{d+\floor{e}+1} \nonumber \\
        &\le& \left(\frac{12dn^4Y\eps^{1/e}}{\sigma}\right)^{d+\floor{e}+1} \label{eqn:kmubCloseCenters1}.
    \end{eqnarray}
    This bound can be proven as follows:
    Arthur and Vassilvitskii~\cite[Lemma~5.8]{ArthurVassilvitskii:ICP:2009} have shown
    that we can
    approximate $H$ and $H_0$ by hyperplanes $\tilde{H}$ and $\tilde{H_0}$
    that pass through $d$ points from $\points$ exactly such that
    any point $x\in\points$ within distance $L$ of $H$ or $H_0$
    has a distance of at most $2dL$ from $\tilde{H}$ or $\tilde{H_0}$, respectively.
    A union bound over all choices for these $2d$ points and the remaining $t+1-2d$
    points yields the term $n^{t+1}$. Once $\tilde{H}$ and $\tilde{H_0}$ are fixed,
    the probability that a random point is within distance $2dL$ of at least one of
    the hyperplanes is bounded from above by $4dL/\sigma$. Taking into account that the
    remaining $t+1-2d$ points are independent Gaussians yields a final bound of
    $n^{t+1}(4dL/\sigma)^{t+1-2d}$ with $L=3nY\eps^{1/e}$.

    Note that this quantity bounds the probability that there \emph{exists} an iteration with $|A' \cup A| \ge t+1$ satisfying the conditions in the lemma
    statement; it does not apply only to the fixed iteration $I$ that we were considering earlier.

    Next we consider the case $|A' \cup A| \le t$. We must have
    $A' \neq A$ since some point is exchanged between clusters $A$ and $B$ during
    iteration $I$. Consider some fixed $A$ and $A'$, and let $x_0$ be a data point in the symmetric difference of $A$ and $A'$. Then $\mass(A') - \mass(A)$
    can be written as $\sum_{x \in X} c_x \cdot x$ for constants $c_x$ with $|c_{x_0}| \ge \frac{1}{n}$. We consider only the randomness in
    the perturbed position of $x_0$ and allow all other points in $X$ to be fixed adversarially. Then
    $\mass(A') - \mass(A)$ follows a normal distribution with standard deviation at least $\frac{\sigma}{n}$, and hence $\norm{\mass(A') - \mass(A)} \le
    \sqrt{\eps}$ with probability at most $(n\sqrt{\eps}/{\sigma})^d$.
    On the other hand, Lemma~\ref{lem:movement}
    implies that $\norm{\mass(A') - \mass(A)} \le \sqrt{\eps}$ must hold for iteration $I$.
    Otherwise, $I$ would result in a potential drop of at most $\eps$. Now, the total number of
    possible sets $A$ and $A'$ is bounded by $(4n)^t$: we choose $t$ candidate points
    to be in $A \cup A'$ and then for each point, we choose which
    set(s) it is in. Taking a union bound over all possible choices, we see the case $|A' \cup A| \le t$ can occur with a probability of at most
    \begin{eqnarray}
        (4n)^t \cdot \left(\frac{n\sqrt{\eps}}{\sigma}\right)^d
        &=& \left(\frac{4^{3+\floor{e}/d}n^{4+\floor{e}/d}\sqrt{\eps}}{\sigma}\right)^d \label{eqn:kmubCloseCenters2}.
    \end{eqnarray}

    Combining equations \eqref{eqn:kmubCloseCenters1} and \eqref{eqn:kmubCloseCenters2}, we have
    \begin{eqnarray*}
        \Prob[\delta_\eps \le Y\eps^{1/e}]
        &\le& \left(\frac{12dn^4Y\eps^{1/e}}{\sigma}\right)^{d+\floor{e}+1} + \left(\frac{4^{3+\floor{e}/d}n^{4+\floor{e}/d}\sqrt{\eps}}{\sigma}\right)^d.
    \end{eqnarray*}
    Note that $d + \floor{e} + 1 \ge e$ and $d \ge 2$, so we can reduce exponents according to Fact~\ref{fact:kmubProbExponents}:
    \begin{align*}
        \Prob[\delta_\eps \le Y\eps^{1/e}]
        &\le \left(\frac{12dn^4Y\eps^{1/e}}{\sigma}\right)^e + \left(\frac{4^{3+\floor{e}/d}n^{4+\floor{e}/d}\sqrt{\eps}}{\sigma}\right)^2\\
        &\le \eps \cdot \left(\frac{12dn^4Y}{\sigma}\right)^e + \eps \cdot \left(\frac{4^{6+e}n^{8+e}}{\sigma^2}\right) & \text{since $d \ge 2$}\\
        &\le \eps \cdot \left(\frac{12n^5Y}{\sigma}\right)^e + \eps \cdot \left(\frac{4^4n^5}{\sigma}\right)^e & \text{since $d \le n$, $e \ge 2$ and $\sigma \le 1$}\\
        &\le \eps \cdot \left(\frac{O(1) \cdot n^5Y}{\sigma}\right)^e.\qedhere
    \end{align*}
\end{proof}

\section{Transition Blueprints}

Our smoothed analysis of $k$-means is based on the potential function $\Psi$. If
$\points\subseteq\cube$, then after the first iteration, $\Psi$ will always be
bounded from above by a polynomial in $n$ and $1/\sigma$. Therefore,
$k$-means terminates quickly if we can lower-bound the drop in $\Psi$ during each
iteration.
So what must happen for a $k$-means iteration to result in a small potential
drop? Recall that any iteration consists of two
distinct phases: assigning points to centers, and then recomputing
center positions. Furthermore, each phase can only decrease the
potential. According to Lemmas~\ref{lem:epsdeltadrop} and~\ref{lem:movement},
an iteration can only result in a small potential drop if none of the
centers move significantly and no point is reassigned that
has a significant distance to the corresponding bisector.
The previous analyses~\cite{ArthurVassilvitskii:ICP:2009,MantheyRoeglin:kMeans:2009}
essentially use a union bound over all possible iterations to show that
it is unlikely that there is an iteration in which none of these events happens.
Thus, with high probability, we get a significant potential drop in every iteration.
As the number of possible iterations can only be bounded by $n^{3kd}$, these
union bounds are quite wasteful and yield only super-polynomial
bounds.

We resolve this problem by introducing the notion of
\emph{transition blueprints}. Such a blueprint is a description
of an iteration of $k$-means that \emph{almost} uniquely determines
everything that happens during the iteration. In particular, one blueprint
can simultaneously cover many similar iterations, which will dramatically
reduce the number of cases that have to be considered in the union bound.
We begin with the notion of a transition graph, which is part of a
transition blueprint.

\begin{definition}
Given a $k$-means iteration, we define its \emph{transition graph} to
be the labeled, directed multigraph with one vertex for each cluster, and
with one edge $(\cluster i,\cluster j)$ with label $x$ for each data point $x$
switching from cluster $\cluster i$ to cluster $\cluster j$.
\end{definition}

We define a vertex in a transition graph to be \emph{balanced} if its in-degree
is equal to its out-degree. Similarly, a cluster is balanced during a $k$-means
iteration if the corresponding vertex in the transition graph is balanced.

To make the full blueprint, we also require information on approximate positions
of cluster centers. We will see below that for an unbalanced cluster this
information can be deduced from the data points that change to or from this
cluster. For balanced clusters we turn to brute force: We tile the hypercube
$\cube$ with a lattice $L_\eps$, where consecutive points are are at a distance of
$\sqrt{n\eps/d}$ from each other,
and choose one point from $L_\eps$ for every balanced cluster.

\begin{definition}
An $(m,b,\eps)$ \emph{transition blueprint} $\blueprint$ consists of a
weakly connected transition graph $G$ with $m$ edges and $b$ balanced
clusters, and one lattice point in $L_{\eps}$ for each balanced
cluster in the graph. A $k$-means iteration is said to \emph{follow}
$\blueprint$ if $G$ is a connected component of the iteration's
transition graph and if the lattice point selected for each
balanced cluster is within a distance of at most $\sqrt{n\eps}$ of the cluster's
actual center position.
\end{definition}
\noindent If $\points\subseteq\cube$, then by the Pythagorean theorem,
every cluster center must
be within distance $\sqrt{n\eps}$ of some point in $L_{\eps}$.
Therefore, every $k$-means iteration follows at least one transition
blueprint.

As $m$ and $b$ grow, the number of valid $(m,b,\eps)$
transition blueprints grows exponentially, but the probability
of failure that we will prove in the following section decreases equally
fast, making the union bound possible. This is what we gain by
studying transition blueprints rather than every possible
configuration separately.

For an unbalanced cluster $\calC$
that gains the points $A \subseteq \points$ and loses the
points $B \subseteq \points$ during the considered iteration,
the \emph{approximate center of $\calC$} is defined as
\[
  \frac{|B| \mass(B) - |A| \mass(A)}{|B| - |A|}\:.
\]
If $\calC$ is balanced, then the approximate center of $\calC$ is
the lattice point specified in the transition blueprint.
The \emph{approximate bisector of $\cluster i$ and $\cluster j$} is the
bisector of the approximate centers of $\cluster i$ and $\cluster j$.
Now consider a data point $x$ switching from some cluster $\cluster{i}$ to some
other cluster $\cluster{j}$. We say the \emph{approximate bisector corresponding to
$x$} is the hyperplane bisecting the approximate centers of $\cluster{i}$ and
$\cluster{j}$. Unfortunately, this definition applies only if $\cluster{i}$ and
$\cluster{j}$ have distinct approximate centers, which is not necessarily the case
(even after the random perturbation). We will call a blueprint
\emph{non-degenerate} if the approximate bisector is in fact well defined for each data
point that switches clusters.
The intuition is that, if one actual cluster center is far away
from its corresponding approximate center, then during
the considered iteration the cluster center must move significantly,
which causes a potential drop according to Lemma~\ref{lem:movement}.
Otherwise, the approximate bisectors are close to the
actual bisectors and we can show that it is unlikely that all
points that change their assignment are close to their corresponding
approximate bisectors. This will yield a potential drop according
to Lemma~\ref{lem:epsdeltadrop}.

The following lemma formalizes what we mentioned above:
If the center of an unbalanced cluster is far away
from its approximate center, then this causes a potential drop in the
corresponding iteration.
\begin{lemma}
\label{lem:farfromapproximate}
Consider an iteration of $k$-means in which a cluster $\calC$ gains
a set $A$ of points and loses a set $B$ of points with
$|A| \neq |B|$.
If $\bigl\| \mass(\calC) - \frac{|B| \mass(B) - |A|\mass(A)}{|B| - |A|}\bigr\| \geq \sqrt{n\eps}$,
then the potential decreases by at least $\eps$.
\end{lemma}

\begin{proof}
Let $\calC'=(\calC\setminus B)\cup A$ denote the cluster
after the iteration.
According to Lemma~\ref{lem:movement},
the potential drops in the considered iteration by at least
\begin{align*}
  |\calC'|\cdot \norm{\mass(\calC')-\mass(\calC)}^2
  & = (|\calC|+|A|-|B|)\left\| \frac{|\calC|\mass(\calC)+|A|\mass(A)-|B|\mass(B)}{|\calC|+|A|-|B|}-\mass(\calC)\right\|^2\\
  & = \frac{\bigl||B|-|A|\bigr|}{|\calC|+|A|-|B|}\left\| \mass(\calC) - \frac{|B|\mass(B)-|A|\mass(A)}{|B|-|A|}\right\|^2
    \ge \frac{(\sqrt{n\eps})^2}{n}\:.\qedhere
\end{align*}
\end{proof}

Now we show that we get a significant potential drop
if a point that changes its assignment is far from
its corresponding approximate bisector.
Formally, we will be studying the following quantity
$\Lambda(\blueprint)$.

\begin{definition}
    Fix a non-degenerate $(m, b, \eps)$-transition blueprint $\blueprint$. Let $\Lambda(\blueprint)$ denote the maximum distance
    between a data point in the transition graph of $\blueprint$ and its corresponding approximate bisector.
\end{definition}
\begin{lemma} \label{lemma:kmubApproximateBisectors}
    Fix $\eps \in [0, 1]$ and a non-degenerate $(m, b, \eps)$-transition blueprint $\blueprint$. If there exists an iteration that
    follows $\blueprint$ and that results in a potential drop of at most $\eps$, then
    \begin{eqnarray*}
        \delta_\eps \cdot \Lambda(\blueprint) \le 6D\sqrt{nd\eps}.
    \end{eqnarray*}
\end{lemma}

\begin{proof}
    Fix an iteration that follows $\blueprint$ and that results in a potential drop of
    at most $\eps$. Consider a data point $x$
    that switches between clusters $\cluster{i}$ and $\cluster{j}$ during this iteration. Let $p$ and $q$ denote the center positions of
    these two clusters at the beginning of the iteration, and let $p'$ and $q'$ denote the approximate center positions of the clusters.
    Also let $H$ denote the hyperplane bisecting $p$ and $q$, and let $H'$ denote the hyperplane bisecting
    $p'$ and $q'$.

    We begin by bounding the divergence between the hyperplanes $H$ and $H'$.

    \begin{sublemma} \label{sublemma:kmubApproximateBisectors}
        Let $u$ and $v$ be arbitrary points on $H$. Then, $\dist(v, H') - \dist(u, H') \le \frac{4\sqrt{n\eps}}{\delta_\eps} \cdot \norm{v - u}$.
    \end{sublemma}

    \begin{proof}
        Let $\theta$ denote the angle between the normal vectors of the hyperplanes $H$ and $H'$. We move the vector
        $\overrightarrow{p'q'}$ to become $\overrightarrow{pq''}$ for some point $q''$, which ensures $\angle qpq'' = \theta$.
        Note that $\norm{q'' - q} \leq \norm{q'' - q'} + \norm{q' - q} = \norm{p - p'} + \norm{q' - q} \leq 2\sqrt{n\eps}$
        by Lemma~\ref{lem:farfromapproximate}.

        Let $r$ be the point where the bisector of the angle $\angle qpq''$ hits the segment $\overline{qq''}$. By the sine law,
        we have
        \begin{eqnarray*}
            \sin\left(\frac{\theta}{2}\right)
            &=& \sin(\angle prq) \cdot \frac{\norm{r-q}}{\norm{p-q}}\\
            &\le& \frac{\norm{q'' - q}}{\norm{p-q}}
             \: \le\: \frac{2\sqrt{n\eps}}{\delta_\eps}.
        \end{eqnarray*}
        Let $y$ and $y'$ be unit vectors in the direction $\overrightarrow{pq}$ and $\overrightarrow{p'q'}$, respectively,
        and let $z$ be an arbitrary point on $H'$. Then,
        \begin{align*}
            \dist(v, H') - \dist(u, H')
            &= |(v-z) \cdot y'| - |(u-z) \cdot y'|\\
            &\le |(v-u) \cdot y'| & \text{by the triangle inequality} \\
            &= |(v-u) \cdot y + (v-u) \cdot (y'-y)| \\
            &= |(v-u) \cdot (y'-y)| & \text{since $u,v \in H$} \\
            &\le \norm{v-u} \cdot \norm{y'-y}.
        \end{align*}
        Now we consider the isosceles triangle formed by the normal vectors $y$ and $y'$. The angle between $y$ and $y'$ is $\theta$.
        Using the sine law again, we get
        \begin{eqnarray*}
            \norm{y'-y} = 2 \cdot \sin\left(\frac{\theta}{2}\right) \le \frac{4\sqrt{n\eps}}{\delta_\eps},
        \end{eqnarray*}
        and the claim follows.
    \end{proof}

    We now continue the proof of Lemma~\ref{lemma:kmubApproximateBisectors}.
    Let $h$ denote the foot of the perpendicular from $x$ to $H$,
    and let $m = \frac{p+q}{2}$. Then,
    \begin{eqnarray}
        \dist(x,H')
        &\le& \|x - h\| + \dist(h, H') \notag\\
        &=& \dist(x, H) + \dist(m, H') + \dist(h, H') - \dist(m,H') \notag\\
        &\le& \dist(x, H) + \dist(m, H') + \frac{4\sqrt{n\eps}}{\delta_\eps} \cdot \norm{h-m}, \label{eqn:kmubApproximateBisectors}
    \end{eqnarray}
    where the last inequality follows from Claim~\ref{sublemma:kmubApproximateBisectors}. By
    Lemma~\ref{lem:epsdeltadrop}, we know that
    the total potential drop during the iteration is at least
    $2 \cdot \norm{p - q} \cdot \dist(x, H)$. However, we assumed that this drop
    was at most $\eps$, so we therefore have $\dist(x, H) \le \frac{\eps}{2\delta_\eps}$. Also, by Lemma~\ref{lem:farfromapproximate},
    \begin{eqnarray*}
        \dist(m, H') \leq \normBig{\frac{p'+q'}2 - \frac{p+q}2} \le \frac{1}{2} \cdot \norm{p'-p} + \frac{1}{2} \cdot \norm{q'-q} \le \sqrt{n\eps}.
    \end{eqnarray*}
    Furthermore, $\norm{h-m} \le \norm{m-x} \le D\sqrt{d}$ since $h-m$ is perpendicular to $x-h$ and $m-x$ lies in the hypercube
    $[-D/2,D/2]^d$. Plugging these bounds into equation \eqref{eqn:kmubApproximateBisectors}, we have
    \begin{align*}
        \dist(x,H')
        &\le \frac{\eps}{2\delta_\eps} + \sqrt{n\eps} + \frac{4D\sqrt{nd\eps}}{\delta_\eps} \notag\\
        &\le \sqrt{n\eps} + \frac{5D\sqrt{nd\eps}}{\delta_\eps} & \text{since $\eps \le 1$}\\
        &\le \frac{6D\sqrt{nd\eps}}{\delta_\eps} & \text{since $\delta_\eps \le \norm{p - q} \le D\sqrt{d}$.}
    \end{align*}
    This bound holds for all data points $x$ that switch clusters, so the lemma follows.
\end{proof}

\section{Analysis of Transition Blueprints}
\label{sec:smoothed}

Let $\Delta$ denote the smallest improvement of the potential $\Psi$ made by any
sequence of three consecutive iterations of the $k$-means method.
In the following, we will define and analyze some variables $\Delta_i$ such
that $\Delta$ can be bounded from below by the minimum of the $\Delta_i$.
These random variables are essentially a case analysis covering
different types of transition graphs. The first five cases deal
with special types of blueprints that require separate attention and do
not fit into the general framework of case six.
The sixth and most involved case (Section~\ref{ssec:other})
deals with general blueprints.

When analyzing these random variables, we will ignore the case that a cluster
can lose all its points in one iteration. If this happens, then $k$-means
continues with one cluster less, which can happen only $k$ times. Since the potential
$\Psi$ does not increase even in this case, this gives only an additive term of $k$
to our analysis.

In the lemmas in this section, we do not specify the parameters $m$
and $b$ when talking about transition blueprints. When
we say \emph{an iteration follows a blueprint with some property $P$},
we mean that there are parameters $m$ and $b$ such that the iteration follows
an $(m,b,\eps)$ transition blueprint with property $P$,
where $\eps$ will be clear from the context.

\subsection{Balanced Clusters of Small Degree}
\label{sec:smallbalance}

\begin{lemma}\label{lemma:smallBalanced}
Fix $\eps\ge0$ and a constant $\const_1\in\NN$.
Let $\Delta_1$ denote the smallest improvement made by any iteration that follows a
blueprint with a balanced non-isolated node of in- and outdegree
at most $\const_1d$. Then,
\[
  \Pr{\Delta_1\le \eps}
  \le \eps\cdot\left(\frac{n^{4\const_1+1}}{\sigma^2}\right).
\]
\end{lemma}
\begin{proof}
We denote the balanced cluster of in- and outdegree
at most $\const_1d$ by $\calC$.
If the center of $\calC$ moves by $\delta$, then the potential drops
by at least $|\calC|\delta^2$.
Hence, $\Delta_1$ can only be smaller than $\eps$ if the center of $\calC$
moves by at most
$\sqrt{\eps/|\calC|}$ during the considered iteration.
Let $A$ and $B$ with $|A|=|B|\leq \const_1d$
be the sets of data points corresponding to the
incoming and outgoing edges of $\calC$, respectively.
If $|A|\mass(A)$ and $|B|\mass(B)$ differ by at least $\sqrt{n\eps} \geq \sqrt{|\calC|\eps}$,
then the cluster center moves by at least $\sqrt{\eps/|\calC|}$ as shown
by the following reasoning:
Let $c$ be the center of mass of the points that belong to $\calC$
at the beginning of the iteration and remain in $\calC$ during the
iteration.
Then the center of mass of $\calC$ moves from
$\frac{(|\calC| -|A|) c + |A| \mass(A)}{|\calC|}$ to
$\frac{(|\calC| -|A|) c + |B| \mass(B)}{|\calC|}$. Since $|A| = |B|$, these two locations
differ by
\[
\normBig{\frac{|B|\mass(B) - |A|\mass(A)}{|\calC|}}
\geq \sqrt{\eps/|\calC|}\:.
\]
By Lemma~\ref{lem:movement}, this causes a potential drop of
at least $|\calC| (\sqrt{\eps/|\calC|})^2 = \eps$.
The random variable $|A|\mass(A)$ is a Gaussian random variable with a
standard deviation of $\sqrt{|A|}\sigma  \geq \sigma$.
If the points of $B$ are fixed arbitrarily, then $|A|\mass(A)$ has to assume
a position within distance $\sqrt{n\eps}$ of $|B|\mass(B)$ for the iteration
to make an improvement of at most $\eps$.

Now we apply a union bound over all possible choices of $A$ and $B$.
We can assume that both $A$ and $B$ contain exactly $\const_1d$ points.
Otherwise, we can pad them by adding the same
points to both of them, which does not affect the analysis.
Hence, the number of choices is bounded by $n^{2\const_1d}$, and we get
\begin{align*}
  \Pr{\Delta_1\le \eps} &
  \le \Pr{\exists A,B,|A|=|B|=\const_1d \colon
    \bigl\||A|\mass(A)-|B|\mass(B)\bigr \|\le \sqrt{n\eps}}\\
  & \le n^{2\const_1d}\left(\frac{\sqrt{n \eps}}{\sigma}\right)^d
  \le \left(\frac{n^{2\const_1+\frac12}\sqrt{\eps}}{\sigma}\right)^d.
\end{align*}
Using Fact~\ref{fact:kmubProbExponents} and $d\ge 2$ concludes the proof.
\end{proof}

\subsection{Nodes of Degree One}
\label{ssec:degreeone}

\begin{lemma}
Fix $\eps\in[0,1]$.
Let $\Delta_2$ denote the smallest improvement made by any iteration
that follows a blueprint with a node of degree $1$. Then,
\[
  \Pr{\Delta_2\le \eps} \le
     \eps\cdot \left(\frac{O(1)\cdot n^{11}}{\sigma^2}\right).
\]
\end{lemma}
\begin{proof}
Assume that a point $x$ switches from cluster $\cluster 1$ to cluster $\cluster
2$, and let $c_1$ and $c_2$ denote the positions of the cluster centers at the
beginning of the iteration. Let $\nu$ be the distance between $c_1$ and $c_2$. Then
$c_2$ has a distance of $\nu/2$ from the bisector of $c_1$ and $c_2$, and the
point $x$ is on the same side of the bisector as $c_2$.

If $\cluster 1$ has only one edge, then the center of cluster $\cluster 1$
moves during this iteration by at least $\frac{\nu}{2(|\cluster 1|-1)}$,
where $|\cluster 1|$ denotes the number of points belonging to $\cluster 1$ at
the beginning of the iteration:
the point $x$ has a distance of at least $\nu/2$ from $c_1$, which yields
a movement of
\[
  \left\|c_1 - \frac{c_1|\cluster 1| - x}{|\cluster 1| -1}\right\|
  = \left\| \frac{c_1 - x}{|\cluster 1| -1}\right\| \geq \frac{\nu}{2(|\cluster 1| -1)} \:.
\]
Hence, the potential drops by at least
$(|\cluster 1|-1)\big(\frac{\nu}{2|\cluster 1|-2}\big)^2\ge \frac{\nu^2}{4|\cluster 1|}
\geq \frac{\nu^2}{4n}$.

If $\cluster 2$ has only one edge, then let
$\alpha$ be the distance of the point $x$ to the bisector of $c_1$ and $c_2$.
By reassigning the point, we get a potential drop of $2\alpha\nu$. Additionally,
$\|x-c_2\| \geq |\nu/2 - \alpha|$. Thus,
$\cluster 2$ moves by at least
\[
  \left\| c_2 - \frac{c_2|\cluster 2| + x}{|\cluster 2| +1}\right\|
  \geq \left\|\frac{c_2 - x}{|\cluster 2| + 1|}\right\|
  \geq \frac{|\nu/2 - \alpha|}{|\cluster 2|+1} \: .
\]
This causes a potential drop
of at least $(|\cluster 2|+1)(\nu/2-\alpha)^2/(|\cluster 2|+1)^2
= (\nu/2-\alpha)^2/(|\cluster 2|+1) \geq (\nu/2 - \alpha)^2/n$.
Hence, the potential drops by at least
\[
   2\alpha\nu+\frac{(\nu/2-\alpha)^2}{n}
   \ge \frac{(\nu/2+\alpha)^2}{n}
   \ge \frac{\nu^2}{4n}\:.
\]

We can assume $\nu \geq \delta_\eps$ since $\delta_\eps$ denotes the closest distance
between any two simultaneous centers in iterations leading to a potential drop of
at most $\eps$.
To conclude the proof, we combine the two cases: If $\cluster 1$ has only one edge,
the potential drop can only be bounded from above by $\eps$ if
$\eps\ge \frac{\nu^2}{4n}\geq \frac{\delta_\eps^2}{4n}$.
Similarly, if $\cluster 2$ has only one edge,
the potential drop can only be bounded from above by $\eps$ if
$\eps\ge \frac{\delta_\eps^2}{4n}$.
Hence, Lemma~\ref{lemma:kmubCloseCenters} yields
\[
  \Pr{\Delta_2\le \eps}
   \le \Pr{\delta_\eps^2/(4n) \le \eps} =
  \PrB{\delta_\eps\le \sqrt{4n\eps}}
    \le \eps\cdot \left(\frac{O(1)\cdot n^{11}}{\sigma^2}\right).\qedhere
\]
\end{proof}

\subsection{Pairs of Adjacent Nodes of Degree Two}
\label{ssec:sixpairs}

Given a transition blueprint, we now look at pairs of adjacent nodes of degree 2.
Since we have already dealt with the case of balanced clusters of small degree
(Section~\ref{sec:smallbalance}), we can assume that the nodes involved are
unbalanced.
This means that one cluster of the pair gains two points while the other
cluster of the pair loses two points.

\begin{lemma}
\label{lem:Delta3}
Fix $\eps\in[0,1]$.
Let $\Delta_3$ denote the smallest improvement made by any iteration that follows
a non-degenerate blueprint with at least three disjoint pairs of adjacent unbalanced nodes
of degree $2$. Then,
\[
   \Pr{\Delta_3\le\eps} \le \eps\cdot\left(\frac{O(1)\cdot n^{30}}{\sigma^6}\right).
\]
\end{lemma}

\begin{proof}
    Fix a transition blueprint $\blueprint$ containing at least 3 disjoint pairs of adjacent unbalanced degree-two nodes. We first bound
    $\Prob[\Lambda(\blueprint) \le \lambda]$. For $i = 1,2,3$, let $a_i$, $b_i$, and $c_i$ denote the data points corresponding to the
    edges in the $i^\textrm{th}$ pair of adjacent degree-two nodes, and assume without loss of generality that $b_i$ corresponds
    to the inner edge (the edge that connects the pair of degree-two nodes).

    Let $\cluster{i}$ and ${\cal C}_i'$ be the clusters corresponding to one such pair of nodes. Since $\cluster{i}$ and ${\cal C}_i'$ are unbalanced,
    we can further assume without loss of generality that ${\cal C}_i$ loses both data points $a_i$ and $b_i$ during the iteration, and
    ${\cal C}_i'$ gains both data points $b_i$ and $c_i$.

    Now, $\cluster{i}$ has its approximate center at
    $p_i = \frac{a_i+b_i}{2}$ and ${\cal C}_i'$ has its approximate center at
    $q_i = \frac{b_i+c_i}{2}$. Since
    $\blueprint$ is non-degenerate, we know $p_i \neq q_i$ and hence $a_i \neq c_i$. Let $H_i$
    denote the hyperplane bisecting $a_i$ and $c_i$, and let $H'_i$ denote the hyperplane bisecting $p_i$ and $q_i$. Since $H_i$ is the image of
    $H'_i$ under a dilation with center $b_i$ and scale $2$, we have
    \begin{eqnarray}
        \Lambda(\blueprint) \ge \max_i\big(\dist(b_i, H'_i)\big) =
        \frac{\max_i\big(\dist(b_i, H_i)\big)}{2}.
    \end{eqnarray}
    All three pairs of adjacent degree-two nodes are disjoint, so we know $b_i$
    is distinct from $b_j$ for $j \neq i$ and distinct from
    $a_j$ and $c_j$ for all $j$. This implies the position of $b_i$ is independent of $b_j$ for $j \neq i$, and it is also independent of the
    position and orientation of $H_j$ for all $j$. In particular, the quantities
    $\dist(b_i, H_i)$ follow independent one-dimensional normal
    distributions with standard deviation $\sigma$. Therefore, for any $\lambda \ge 0$, we have
    \[
        \Prob\left[\Lambda(\blueprint) \le \lambda\right]
        \le \Prob\left[\max_i \Big(\dist(b_i, H_i)\Big) \le 2\lambda\right]
        \le \left(\frac{2\lambda}{\sigma}\right)^3.
    \]

    Let $\mathbb{B}$ denote the set of non-degenerate transition blueprints containing
    at least three disjoint pairs of unbalanced degree-two nodes.
    The preceding analysis of $\Prob[\Lambda(\blueprint) \le \lambda]$ depends only on
    $\{a_i, b_i, c_i\}$ so we can use a union bound over
    all choices of $\{a_i, b_i, c_i\}$ as follows:
    \begin{eqnarray}
        \Prob\Big[\exists \blueprint \in \mathbb{B} \,\Big|\,
        \Lambda(\blueprint) \le \lambda\Big]
        \le n^9 \cdot \left(\frac{2\lambda}{\sigma}\right)^3
        = \left(\frac{2n^3\lambda}{\sigma}\right)^3 \label{eqn:kmubCase3}.
    \end{eqnarray}

    Now, Lemma~\ref{lemma:kmubApproximateBisectors} yields that if an iteration
    can follow a blueprint $\blueprint$ and result in a potential drop of at most
    $\eps$, then $\delta_\eps \cdot \Lambda(\blueprint) \le 6D\sqrt{nd\eps}$. We
    must therefore have either $\delta_\eps \le \eps^{1/6}$
    or $\Lambda(\blueprint) \le 6D\sqrt{nd} \cdot \eps^{1/3}$. We bound the probability that
    this can happen using Lemma~\ref{lemma:kmubCloseCenters} and equation~\eqref{eqn:kmubCase3}:
    \begin{eqnarray*}
        \Pr{\Delta_3\le\eps}
        &\le& \Prob\left[\delta_\eps \le \eps^{1/6}\right] + \Prob\left[\exists \blueprint \in \mathbb{B} \Big| \Lambda(\blueprint) \le 6D\sqrt{nd} \cdot \eps^{1/3}\right]\\
        &\le& \eps \cdot \left(\frac{O(1) \cdot n^5}{\sigma}\right)^6 + \eps \cdot \left(\frac{12Dn^3\sqrt{nd}}{\sigma}\right)^3\\
        &=& \eps \cdot \left(\frac{O(1) \cdot n^{30}}{\sigma^6}\right),
    \end{eqnarray*}
    since $D = \sqrt{90kd \cdot \ln(n)}$, $\sigma \le 1$, and $d, k \le n$.
\end{proof}

\subsection{Blueprints with Constant Degree}
\label{ssec:constant}

Now we analyze iterations that follow blueprints in which every node
has constant degree. It might happen that
a single iteration does not yield a significant improvement
in this case. But we get a significant improvement
after three consecutive iterations of this kind.
The reason for this is that
during three iterations one cluster must
assume three different configurations.
One case in the previous
analyses~\cite{ArthurVassilvitskii:ICP:2009,MantheyRoeglin:kMeans:2009}
is iterations in which every cluster exchanges
at most $O(dk)$ data points with other clusters.
The case considered in this section is similar, but
instead of relying on the somewhat cumbersome notion of \emph{key-values}
used in the previous analyses,
we present a simplified and more intuitive analysis here,
which also sheds more light on the previous analyses.

We define an \emph{epoch} to be a sequence of
consecutive iterations in which no cluster center assumes more than two different
positions. Equivalently, there are at most two different sets $\cluster{i}',
\cluster{i}''$ that every cluster $\cluster i$ assumes. Arthur and
Vassilvitskii~\cite{ArthurVassilvitskii:ICP:2009} used the obvious upper bound of
$2^k$ for the length of an epoch (the term \emph{length} refers
to the number of iterations in the sequence).
This upper bound has been improved to two.
By the definition of length
of an epoch, this means that after at most three iterations, either $k$-means
terminates or one cluster assumes a third configuration.

\begin{lemma}[\mbox{Manthey, R\"oglin~\cite[Lemma 4.1]{MantheyRoeglin:kMeans:2009}}]
\label{lem:epochelength}
The length of any epoch is at most two.
\end{lemma}

For our analysis, we introduce the notion of \emph{$(\eta,c)$-coarseness}.
In the following, $\xor$ denotes the symmetric difference of two sets.

\begin{definition}
We say that
$\points$ is \emph{$(\eta,c)$-coarse} if for any pairwise distinct subsets
$\cluster 1$, $\cluster 2$, and $\cluster 3$ of $\points$ with $|\cluster{1}\xor\cluster{2}|\le c$ and
$|\cluster{2}\xor\cluster{3}|\le c$, either
$\norm{\mass(\cluster 1)-\mass(\cluster 2)}>\eta$ or
$\norm{\mass(\cluster 2)-\mass(\cluster 3)}>\eta$.
\end{definition}

According to Lemma~\ref{lem:epochelength}, in every sequence of
three consecutive iterations, one cluster assumes three different
configurations. This yields the following lemma.
\begin{lemma}
\label{lem:deltasparsedrop}
Assume that $\points$ is $(\eta, c)$-coarse and consider a sequence
of three consecutive iterations. If in each of these iterations every cluster
exchanges at most $c$ points, then the potential
decreases by at least $\eta^2$.
\end{lemma}
\begin{proof}
According to Lemma~\ref{lem:epochelength}, there is one cluster that assumes
three different configurations $\cluster 1$, $\cluster 2$, and $\cluster 3$ in
this sequence. Due to the assumption in the lemma, we have $|\cluster 1\xor
\cluster 2|\le c$ and $|\cluster 2\xor \cluster 3|\le c$.
Hence, due to the definition of $(\eta, c)$-coarseness, we have
$\norm{\mass(\cluster i)-\mass(\cluster{i+1})}>\eta$ for
one $i\in\{1,2\}$. Combining this with Lemma~\ref{lem:movement} concludes the proof.
\end{proof}

\begin{lemma}
\label{lem:NotCoarse}
For $\eta\ge 0$, the probability that $\points$ is not $(\eta,c)$-coarse is at
most $(7n)^{2c} \cdot (2nc\eta/\sigma)^d$.
\end{lemma}
\begin{proof}
Given any sets $\cluster 1$, $\cluster 2$, and $\cluster 3$ with
$|\cluster{1}\xor\cluster{2}|\le c$ and
$|\cluster{2}\xor\cluster{3}|\le c$, we can write $\cluster i$, for
$i\in\{1,2,3\}$, uniquely as the disjoint union of a common ground set
$A\subseteq\points$ with a set $B_i\subseteq\points$ with $B_1\cap B_2\cap B_3=\emptyset$.
Furthermore,
\[
        \quad B_1 \cup B_2 \cup B_3
        = (C_1 \cup C_2 \cup C_3) \setminus A
        = (C_1 \xor C_2) \cup (C_2 \xor C_3),
\]
so $|B_1 \cup B_2 \cup B_3| = |(C_1 \xor C_2) \cup (C_2 \xor C_3)| \le 2c$.

We perform a union bound over all choices for the sets
$B_1$, $B_2$, and $B_3$. The number of choices for these sets is bounded from above by
$7^{2c}\binom{n}{2c}\le (7n)^{2c}$: we choose $2c$ candidate points to be in $B_1 \cup B_2 \cup B_3$ and then
    for each point, we choose which set(s) it is in.
We assume in the following that the sets $B_1$, $B_2$, and $B_3$ are fixed.
For $i\in\{1,2\}$, we can write $\mass(\cluster{i})-\mass(\cluster{i+1})$ as
\begin{equation}\label{eqn:ClusterDifference}
   \left(\frac{|A|}{|A|+|B_i|}-
   \frac{|A|}{|A|+|B_{i+1}|}\right)\cdot\mass(A)
   +\frac{|B_i|}{|A|+|B_{i}|}\cdot\mass(B_i)
   -\frac{|B_{i+1}|}{|A|+|B_{i+1}|}\cdot\mass(B_{i+1})\:.
\end{equation}

Let us first consider the case that
we have $|B_i|=|B_{i+1}|$ for
one $i\in\{1,2\}$. Then
$\mass(\cluster{i})-\mass(\cluster{i+1})$ simplifies to
\[
   \frac{|B_i|}{|A|+|B_{i}|}\cdot\left(\mass(B_i)
   -\mass(B_{i+1})\right)
   = \frac{1}{|A|+|B_{i}|}\cdot\left(
   \sum_{x\in B_i\setminus B_{i+1}}x
   -\sum_{x\in B_{i+1}\setminus B_{i}}x\right)
   \:.
\]
Since $B_i\neq B_{i+1}$, there exists a point $x\in B_i\xor B_{i+1}$. Let us assume
without loss of generality that $x\in B_i\setminus B_{i+1}$ and that the positions of all points

in $(B_i\cup B_{i+1})\setminus\{x\}$ are fixed arbitrarily. Then the event that
$\norm{\mass(\cluster i)-\mass(\cluster{i+1})}\le \eta$ is equivalent to the
event that $x$ lies in a fixed hyperball of radius $(|A|+|B_{i}|)\eta \le
n\eta $. Hence, the probability is bounded from above by $(n\eta/\sigma)^d\le
(2nc\eta/\sigma)^d$.

Now assume that $|B_1|\neq|B_2|\neq|B_3|$. For $i \in \{1,2\}$, we set
\[
r_i = \left(\frac{|A|}{|A|+|B_i|}-
   \frac{|A|}{|A|+|B_{i+1}|}\right)^{-1}
   = \frac{(|A|+|B_{i}|)\cdot(|A|+|B_{i+1}|)}{|A|\cdot(|B_{i+1}|-|B_{i}|)}
\]
and
\[
Z_i =
\frac{|B_{i+1}|}{|A|+|B_{i+1}|}\cdot\mass(B_{i+1})
-\frac{|B_i|}{|A|+|B_{i}|}\cdot\mass(B_i)\:.
\]
According to~\eqref{eqn:ClusterDifference}, the event
$\norm{\mass(\cluster{i})-\mass(\cluster{i+1})}\le\eta$
is equivalent to the event that $\mass(A)$ falls into the hyperball with radius
$|r_i|\eta$ and center $r_iZ_i$. Hence, the event that both
$\norm{\mass(\cluster 1)-\mass(\cluster 2)}\le\eta$ and
$\norm{\mass(\cluster 2)-\mass(\cluster 3)}\le\eta$ can only occur if the
hyperballs $\ball(r_1Z_1,|r_1|\eta)$ and $\ball(r_2Z_2,|r_2|\eta)$
intersect. This event occurs if and only if the centers
$r_1Z_1$ and $r_2Z_2$ have a
distance of at most $(|r_1|+|r_2|)\eta$ from each other. Hence,
\[
    \Prob\bigl[(\norm{\mass(\cluster 1)-\mass(\cluster 2)}\le\eta)
      \wedge (\norm{\mass(\cluster 2)-\mass(\cluster 3)}\le\eta)\bigr]
    \le \PrB{\norm{r_1Z_1-r_2Z_2}\le (|r_1|+|r_2|)\eta}\:.
\]
After some algebraic manipulations, we can write the vector $r_1Z_1-r_2Z_2$ as
\begin{align*}
  &-\frac{|A|+|B_2|}{|A|\cdot(|B_2|-|B_1|)}\cdot\sum_{x\in B_1}x
  -\frac{|A|+|B_2|}{|A|\cdot(|B_3|-|B_2|)}\cdot\sum_{x\in B_3}x\\
  &+\left(\frac{|A|+|B_1|}{|A|\cdot(|B_2|-|B_1|)}
  +\frac{|A|+|B_3|}{|A|\cdot(|B_3|-|B_2|)}
  \right)\cdot\sum_{x\in B_2}x\:.
\end{align*}

Since $B_1\neq B_3$, there must be an $x\in B_1\xor B_3$. We can assume
that $x\in B_1\setminus B_3$. If $x\notin B_2$, we let an
adversary choose all positions of the points in $B_1\cup B_2\cup
B_3\setminus\{x\}$. Then the event $\norm{r_1Z_1-r_2Z_2}\le(|r_1|+|r_2|)\eta$
is equivalent to $x$ falling into a fixed hyperball of radius
\[
   \left|\frac{|A|\cdot(|B_2|-|B_1|)}{|A|+|B_2|}(|r_1|+|r_2|)\right|\eta
   = \left|(|B_2|-|B_1|)\cdot\left(
     \left|\frac{|A|+|B_1|}{|B_2|-|B_1|}\right|
     +\left|\frac{|A|+|B_3|}{|B_3|-|B_2|}\right|
   \right)\right|\eta \le 2nc \eta\:.
\]
The probability of this event is thus bounded from above by
$(2nc\eta/\sigma)^d$.

It remains to consider the case that $x\in (B_1\cap B_2)\setminus B_3$. Also in
this case we let an adversary choose the positions of the points in $B_1\cup
B_2\cup B_3\setminus\{x\}$. Now the event $\norm{r_1Z_1-r_2Z_2}\le(|r_1|+|r_2|)\eta$
is equivalent to $x$ falling into a fixed hyperball of radius
\[
   \left|\frac{|A|\cdot(|B_3|-|B_2|)}{|A|+|B_2|}(|r_1|+|r_2|)\right|\eta
   = \left|(|B_3|-|B_2|)\cdot\left(
     \left|\frac{|A|+|B_1|}{|B_2|-|B_1|}\right|
     +\left|\frac{|A|+|B_3|}{|B_3|-|B_2|}\right|
   \right)\right|\eta \le 2nc\eta\:.
\]
Hence, the probability is bounded from above by
$(2nc\eta/\sigma)^d$ also in this case.

This concludes the proof because there are at most $(7n)^{2c}$
choices for $B_1$, $B_2$, and $B_3$ and, for every choice, the probability that
both $\norm{\mass(\cluster 1)-\mass(\cluster 2)}\le\eta$ and
$\norm{\mass(\cluster 2)-\mass(\cluster 3)}\le\eta$ is at most
$(2nc\eta/\sigma)^d$.
\end{proof}

Combining Lemmas~\ref{lem:deltasparsedrop} and~\ref{lem:NotCoarse}
immediately yields the following result.
\begin{lemma}
Fix $\eps\ge 0$ and a constant $\const_2\in\NN$.
Let $\Delta_4$ denote the smallest improvement made by any
sequence of three consecutive iterations that follow blueprints
whose nodes all have degree at most $\const_2$.
Then,
\[
  \Pr{\Delta_4\le\eps} \leq \eps \cdot \left(\frac{O(1) \cdot n^{2(\const_2+1)}}{\sigma^2}\right).
\]
\end{lemma}
\begin{proof}
Taking $\eta=\sqrt{\eps}$, Lemmas~\ref{lem:deltasparsedrop} and~\ref{lem:NotCoarse}
immediately give
\[
  \Pr{\Delta_4\le\eps} \leq (7n)^{2\const_2} \cdot \left(\frac{2n\const_2\sqrt{\eps}}{\sigma}\right)^d.
\]
Since $d\ge2$, the lemma follows from Fact~\ref{fact:kmubProbExponents}
and the fact that $\const_2$ is a constant.
\end{proof}

\subsection{Degenerate blueprints}
\label{ssec:degenerateBlueprints}

\begin{lemma} \label{lemma:kmubCase0}
Fix $\eps\in[0,1]$.
Let $\Delta_5$ denote the smallest improvement made by any
iteration that follows a degenerate blueprint.
Then,
    \begin{eqnarray*}
        \Pr{\Delta_5\le\eps} \leq \eps \cdot \left(\frac{O(1) \cdot n^{11}}{\sigma^2}\right).
    \end{eqnarray*}
\end{lemma}

\begin{proof}
    Consider such an iteration. Since the blueprint is degenerate, there must
    exist two clusters $\cluster{i}$ and $\cluster{j}$ that have identical approximate
    centers and that exchange a data point during the iteration. Let $c_i$ and
    $c_j$ denote the actual centers of these clusters at the beginning of the
    iteration. By Lemma~\ref{lem:farfromapproximate}, $\delta_\eps \le
    \norm{c_i - c_j} \le 2\sqrt{n\eps}$. However, we know from
    Lemma~\ref{lemma:kmubCloseCenters} that this occurs with probability at most
    $\eps \cdot (O(1) \cdot n^{5.5}/{\sigma})^2$.
\end{proof}

\subsection{Other Blueprints}
\label{ssec:other}

Now, after having ruled out five special cases, we can analyze
the case of a general blueprint.

\begin{lemma}\label{lemma:Delta6}
Fix $\eps\in[0,1]$.
Let $\Delta_6$ be the smallest improvement made by any iteration whose
blueprint does not fall into any of the previous five categories with
$\const_1=8$ and $\const_2=7$.
This means that we consider only non-degenerate blueprints whose balanced
nodes have in- and
out-degree at least $8d+1$, that do not have nodes of degree one,
that have at most two disjoint pairs of adjacent unbalanced node of degree
2, and that have a node with degree at least 8.
Then,
\[
  \Pr{\Delta_6\le\eps} \le
  \eps \cdot \left(\frac{O(1) \cdot n^{33}k^{30}d^3D^3}{\sigma^6}\right).
\]
\end{lemma}

Proving this lemma requires some preparation.
Assume that the iteration follows a blueprint $\blueprint$ with
$m$ edges and $b$ balanced nodes.
We distinguish two cases: either the center of one unbalanced cluster
assumes a position
that is $\sqrt{n\eps}$ away from its approximate position
or all centers are at most $\sqrt{n\eps}$ far away from their approximate positions.
In the former case the potential drops
by at least $\eps$ according to Lemma~\ref{lem:farfromapproximate}.
If this is not the case, the potential drops if
one of the points is far away from its corresponding approximate bisector
according to Lemma~\ref{lemma:kmubApproximateBisectors}.

The fact that the blueprint does not belong to any of the previous categories
allows us to derive the following upper bound on its number of nodes.

\begin{lemma}\label{NumberNodesBlueprint}
Let $\blueprint$ denote an arbitrary
transition blueprint with $m$ edges and $b$ balanced nodes
in which every node has degree at least two and
every balanced node has degree at least $2d\const_1+2$.
Furthermore, let
there be at most two disjoint pairs of adjacent nodes of degree two
in $\blueprint$, and assume that there is one node with degree
at least $\const_2+1>2$.
Then the number of nodes in $\blueprint$ is bounded from above by
\[
  \begin{cases}
    \frac{5}{6}m-\frac{\const_2-4}{3} & \text{if $b=0$,}\\
    \frac{5}{6}m-\frac{(2\const_1d-1)b-2}{3} & \text{if $b\ge1$.}
  \end{cases}
\]
\end{lemma}

\begin{proof}
Let $A$ be the set of nodes of
degree two, and let $B$ be the set of nodes of higher degree.
We first bound the number of edges between nodes in $A$:
There are at most two disjoint
pairs of adjacent nodes of degree two.
For each of these pairs, we define
its extension to be the longest path of nodes of degree two containing
the pair. We know that none of these extensions can form a cycle as the
transition graph is connected and contains a node of degree $\const_2+1>2$.
There are $\lfloor h/2 \rfloor$ disjoint pairs in an extension consisting of
$h$ nodes. As the extensions contain all edges between nodes of
degree 2, this implies that
the number of edges between vertices in $A$ is at most four.
Let $\mathrm{deg}(A)$ and $\mathrm{deg}(B)$ denote the sum of the
degrees of the nodes in $A$ and $B$, respectively.
The total degree $\mathrm{deg}(A)$ of
the vertices in $A$ is $2|A|$. Hence, there are at least $2|A|-8$ edges
between $A$ and $B$.
Therefore,
\begin{align*}
  2|A|-8 \le \mathrm{deg}(B) & \Rightarrow \, 2|A|-8 \le 2m-2|A| \\
  & \Rightarrow \, |A| \le \frac{1}{2}m+2\:.
\end{align*}

Let $t$ denote the number of nodes.
The nodes in $B$ have degree at least $3$,
there is one node in $B$ with degree at least $\const_2+1$,
and balanced nodes have degree at least $2\const_1d+2$ (and hence, belong to $B$).
Therefore, if $b=0$,
\begin{alignat*}{2}
  && 2m   &\ge  2|A|+3(t-|A|-1)+\const_2+1\\
  \Rightarrow \;&& 2m+|A| &\ge 3t+\const_2-2\\
  \Rightarrow \; && \frac{5}{2}m  &\ge 3t+\const_2-4\:.
\end{alignat*}
If $b\ge1$, then the node of degree at least $\const_2+1$ might be balanced
and we obtain
\begin{alignat*}{2}
  && 2m   &\ge  2|A|+(2\const_1d+2)b+3(t-|A|-b)\\
  \Rightarrow \;&& 2m+|A| &\ge 3t+(2\const_1d-1)b\\
  \Rightarrow \; && \frac{5}{2}m  &\ge 3t+(2\const_1d-1)b-2\:.
\end{alignat*}
The lemma follows by solving these inequalities for $t$.
\end{proof}

We can now continue to bound $\Prob[\Lambda(\blueprint)\le \lambda]$ for a fixed
blueprint $\blueprint$. The previous lemma implies that a relatively large number of
points must switch clusters, and each such point is positioned independently
according to a normal distribution. Unfortunately, the approximate bisectors are
not independent of these point locations, which adds a technical challenge. We
resolve this difficulty by changing variables and then bounding the effect of
this change.
\begin{lemma}
\label{lemma:FixedBlueprint}
For a fixed transition blueprint $\blueprint$
with $m$ edges and $b$ balanced clusters that does not belong to
any of the previous five categories
and for any $\lambda\ge 0$, we have
\[
  \Pr{\Lambda(\blueprint)\le \lambda} \le
  \begin{cases}
  \left(\frac{\sqrt{d}m^2\lambda}{\sigma}\right)^{\frac m6
   + \frac{\const_2-1}3} & \text{if $b=0$,}\\
   \left(\frac{\sqrt{d}m^2\lambda}{\sigma}\right)^{\frac m6
   +\frac{(2\const_1d+2)b-2}3} & \text{if $b\ge 1$.}
  \end{cases}
\]
\end{lemma}
\begin{proof}
We partition the set of edges in the transition graph into \emph{reference edges} and \emph{test edges}.
For this, we ignore the directions of the edges in the transition graph and
compute a spanning tree in the resulting undirected multi-graph. We let an
arbitrary balanced cluster be the root of this spanning tree. If all clusters are
unbalanced, then an arbitrary cluster is chosen as the root. We mark every edge whose
child is an unbalanced cluster as a reference edge. In this way, every unbalanced
cluster $\cluster i$ can be incident to several reference edges.
But we will refer only to the reference edge between $\cluster i$'s parent and
$\cluster i$ as the reference edge associated
with $\cluster i$. Possibly except for the root, every unbalanced cluster is associated
with exactly one reference edge. Observe that in the transition graph,
the reference edge of an unbalanced
cluster $\cluster i$ can either be directed from $\cluster i$ to its parent or
vice versa, as we ignored the directions of the
edges when we computed the spanning tree. From now on, we will again take
into account the directions of the edges.

For every unbalanced cluster $i$ with an associated
reference edge, we define the point $q_i$ as
\begin{equation}\label{eqn:qi}
  q_i = \sum_{x\in A_i}x-\sum_{x\in B_i}x\:,
\end{equation}
where $A_i$ and $B_i$ denote the sets of incoming
and outgoing edges of $\cluster i$, respectively.
The intuition behind this definition is as follows:
as we consider a fixed blueprint $\blueprint$, once $q_i$ is fixed
also the approximate center of cluster $i$ is fixed.
Let $q$ denote the point defined as in~\eqref{eqn:qi} but for the
root instead of cluster $i$.
If all clusters are
unbalanced and $q_i$ is fixed for every cluster except for the root,
then also the value of $q$ is implicitly fixed as $q+\sum q_i=0$.
Hence, once each $q_i$ is fixed, the approximate center of every
unbalanced cluster is also fixed.

Relabeling as necessary,
we assume without loss of generality that the clusters with an associated
reference edge are the clusters $\cluster 1,\ldots, \cluster r$ and that the corresponding
reference edges correspond to the points $p_1,\ldots,p_r$. Furthermore,
we can assume that
the clusters are topologically sorted: if $\cluster{i}$ is a descendant of $\cluster{j}$,
then $i<j$.

Let us now assume that an adversary chooses an arbitrary position
for $q_i$ for every cluster $\cluster i$ with $i\in[r]$.
Intuitively, we will show that regardless of how the transition blueprint
$\blueprint$ is chosen and regardless of how the adversary fixes the positions
of the $q_i$, there is still enough randomness left to conclude that
it is unlikely that all points involved in the iteration are close to their
corresponding approximate bisectors.
We can alternatively
view this as follows: Our random experiment is to choose the
$md$-dimensional Gaussian vector $\bar{p}=(p_1,\ldots,p_m)$,
where $p_1,\ldots,p_m\in\RR^d$ are the points that correspond to the edges in
the blueprint.
For each $i\in[r]$ and $j\in[d]$ let $\bar{b}_{ij}\in\{-1,0,1\}^{md}$
be the vector so that the $j$-th component of $q_i$ can be written as
$\bar{p}\cdot\bar{b}_{ij}$. Then allowing the adversary to fix the
positions of the $q_i$ is equivalent to letting him fix the value
of every dot product $\bar{p}\cdot\bar{b}_{ij}$.

After the positions of the $q_i$ are chosen, we know the location of the approximate
center of every unbalanced cluster. Additionally, the blueprint provides an
approximate center for every balanced cluster.
Hence, we know the positions of all
approximate bisectors. We would like to estimate the probability that
all points $p_{r+1},\ldots,p_m$ have a distance of at most $\lambda$
from their corresponding approximate bisectors.
For this, we further reduce the randomness
and project each point $p_i$ with $i\in\{r+1,\ldots,m\}$
onto the normal vector of its corresponding approximate bisector.
Formally,
for each $i\in\{r+1,\ldots,m\}$, let $h_i$ denote a normal vector to the approximate bisector corresponding to $p_i$, and let
$\bar{b}_{i,1}\in[-1,1]^{md}$ denote the vector such that $\bar{p}\cdot\bar{b}_{i,1} \equiv p_i \cdot h_i$. This means that $p_i$
is at a distance of at most $\lambda$ from its approximate bisector if and only if
$\bar{p}\cdot\bar{b}_{i1}$ lies in some fixed interval $\calI_i$ of
length $2\lambda$. As this event is independent of the other
points $p_j$ with $j\neq i$, the vector $\bar{b}_{i1}$ is a unit vector in the subspace spanned
by the vectors $e_{(i-1)d+1},\ldots,e_{id}$ from the canonical basis.
Let $\calB_i=\{\bar{b}_{i1},\ldots,\bar{b}_{id}\}$ be an orthonormal basis of this
subspace. Let $M$ denote the $(md)\times(md)$ matrix whose columns
are the vectors $\bar{b}_{11},\ldots,\bar{b}_{1d},\ldots,\bar{b}_{m1},\ldots,\bar{b}_{md}$.
Figure~\ref{fig:exampleMatrix} illustrates these definitions.

\begin{figure}
\newcommand{\nM}{0_d}
\centering
\parbox{4cm}{
\includegraphics{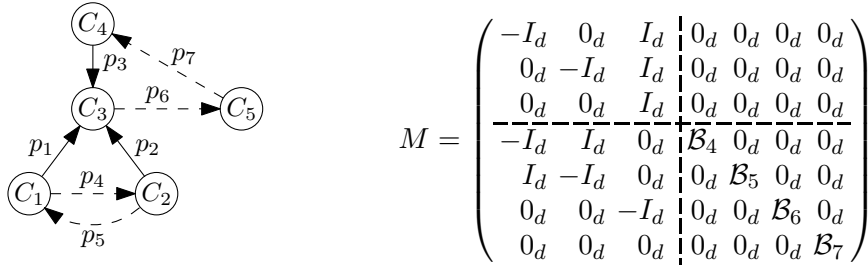}
}
\qquad
\begin{minipage}{7cm}
\[M = \left(\begin{BMAT}{rrr:rrrr}{ccc:cccc}
    -I_d  & \nM  & I_d ~ & \nM & \nM     & \nM & \nM \\
    \nM  & -I_d & I_d ~ & \nM & \nM     & \nM & \nM \\
    \nM  & \nM  & I_d ~ & \nM & \nM     & \nM & \nM \\
    -I_d & I_d  & \nM ~ & \calB_4 & \nM & \nM & \nM \\
    I_d  & -I_d & \nM ~ & \nM     & \calB_5 & \nM & \nM \\
    \nM  & \nM  & -I_d ~& \nM     & \nM & \calB_6 & \nM \\
    \nM  & \nM  & \nM ~ & \nM     & \nM & \nM & \calB_7\\
\end{BMAT}\right)\]
\end{minipage}
\caption{Solid and dashed edges indicate reference and test edges, respectively.
   When computing the spanning tree, the directions of the edges are ignored.
   Hence, reference edges can either be directed from parent to child or vice versa.
   In this example, the spanning tree consists of the edges $p_3$, $p_7$, $p_1$,
   and $p_2$, and its root is $C_4$.
   We denote by $I_d$ the $d\times d$ identity matrix and by $\nM$ the $d\times d$
   zero matrix. The first three columns of $M$ correspond to $q_1$, $q_2$, and
   $q_3$. The rows correspond to the points $p_1,\ldots,p_7$.
   Each block matrix $\mathcal{B}_i$ corresponds to an orthonormal basis of $\RR^d$ and is therefore orthogonal.}
\label{fig:exampleMatrix}
\end{figure}

For $i\in[r]$ and $j\in[d]$, the values of $\bar{p}\cdot\bar{b}_{ij}$ are fixed
by an adversary. Additionally, we allow the adversary to fix the values of
$\bar{p}\cdot\bar{b}_{ij}$ for $i\in\{r+1,\ldots,m\}$ and $j\in\{2,\ldots,d\}$.
All this together defines an $(m-r)$-dimensional affine subspace $U$ of
$\RR^{md}$.
We stress that the subspace $U$ is chosen by the adversary and no assumptions
about $U$ are made. In the following, we will condition on the event that
$\bar{p}=(p_1,\ldots,p_m)$ lies in this subspace.
We denote by $\calF$ the event that
$\bar{p}\cdot\bar{b}_{i1}\in\calI_i$ for all $i\in\{r+1,\ldots,d\}$. Conditioned
on the event that the random vector $\bar{p}$ lies in the subspace $U$, $\bar{p}$
follows an $(m-r)$-dimensional Gaussian distribution with standard deviation
$\sigma$. However, we cannot directly estimate the probability of the event
$\calF$ as the projections of the vectors $\bar{b}_{i1}$ onto the affine subspace $U$
might not be orthogonal.
To estimate the probability of $\calF$, we perform a change of variables. Let
$\bar{a}_1,\ldots,\bar{a}_{m-r}$ be an arbitrary orthonormal basis of the
$(m-r)$-dimensional subspace obtained by shifting $U$ so that it contains the
origin.
Assume for the moment that we had, for each of these vectors $\bar{a}_{\ell}$, an
interval $\calI_{\ell}'$ such that
$\calF$ can only occur if $\bar{p}\cdot\bar{a}_{\ell}\in\calI_{\ell}'$ for every ${\ell}$.
Then we could bound the probability of $\calF$ from above by
$\prod \frac{|\calI_{\ell}'|}{\sqrt{2\pi}\sigma}$ as the $\bar{p}\cdot \bar{a}_{\ell}$ can be
treated as independent one-dimensional Gaussian random variables with standard
deviation $\sigma$ after conditioning on $U$.
In the following, we construct such intervals $\calI_{\ell}'$.

It is important that the vectors $\bar{b}_{ij}$ for $i\in[m]$ and $j\in[d]$
form a basis of $\RR^{md}$.
To see this, let us first have a closer look at the
matrix $M \in \RR^{md \times md}$ viewed as an $m\times m$ block matrix
with blocks of size $d\times d$. From the fact that the reference points are
topologically sorted it follows that the upper left part, which consists
of the first $dr$ rows and columns, is an upper triangular matrix with non-zero
diagonal entries.

As the upper right
$(dr)\times d(m-r)$ sub-matrix of $M$ consists solely of zeros, the determinant
of $M$ is the product of the determinant of the upper left
$(dr)\times (dr)$ sub-matrix and the determinant of the lower right
$d(m-r)\times d(m-r)$ sub-matrix. Both of these determinants can easily be seen
to be different from zero. Hence, also the determinant of $M$ is not
equal to zero, which in turn implies that the vectors $\bar{b}_{ij}$
are linearly independent and form a basis of $\RR^{md}$.

In particular, we can write every $\bar{a}_{\ell}$ as a linear combination of the
vectors $\bar{b}_{ij}$. Let
\[
   \bar{a}_\ell = \sum_{i,j}c_{ij}^\ell\bar{b}_{ij}
\]
for some coefficients $c_{ij}^\ell \in\RR$.
Since the values of $\bar{p}\cdot\bar{b}_{ij}$ are fixed for $i\in[r]$ and $j\in[d]$
as well as for $i\in\{r+1,\ldots,m\}$ and $j\in\{2,\ldots,d\}$, we can write
\[
  \bar{p}\cdot\bar{a}_\ell = \kappa_\ell + \sum_{i=r+1}^m c_{i1}^\ell (\bar{p}\cdot\bar{b}_{i1})
\]
for some constant $\kappa_\ell$ that depends on the fixed values chosen by the adversary.
Let $c_{\max} =\max\{|c_{i1}^l|\mid i>r\}$. The event $\calF$ happens only if, for every $i>r$,
the value of $\bar{p}\cdot\bar{b}_{i1}$ lies in some fixed interval of length $2\lambda$.
Thus, we conclude that $\calF$ can happen only if for every $\ell \in[m-r]$ the value of
$\bar{p}\cdot\bar{a}_\ell$ lies in some fixed interval $\calI_\ell'$ of length
at most $2c_{\max}(m-r)\lambda$. It only remains to bound $c_{\max}$ from above.
For $\ell \in[m-r]$, the vector $c^\ell$ of the coefficients $c_{ij}^\ell$ is obtained as
the solution of the linear system $Mc^\ell=\bar{a}_\ell$.
The fact that the upper right $(dr)\times d(m-r)$ sub-matrix of $M$ consists only
of zeros implies that the first $dr$ entries of $\bar{a}_\ell$ uniquely
determine the first $dr$ entries of the vector $c^\ell$. As $\bar{a}_\ell$ is a unit
vector, the absolute values of all its entries are bounded by $1$.
Now we observe that each row of the matrix $M$ contains at most two non-zero entries in the
first $dr$ columns because every edge in the transition
blueprint belongs to only two clusters. This and a short
calculation shows that the absolute values of the first $dr$ entries of $c$ are
bounded by $r$: The absolute values of the entries $d(r-1)+1,\ldots,dr$ coincide
with the absolute values of the
corresponding entries in $\bar{a}_\ell$ and are thus bounded by $1$.
Given this,
the rows $d(r-2)+1,\ldots,d(r-1)$ imply that the corresponding values in
$\bar{a}_\ell$ are bounded by $2$ and so on.

Assume that the first $dr$ coefficients of $c^\ell$ are fixed to values whose
absolute values are bounded by $r$. This leaves us with a
system $M'(c^\ell)'=\bar{a}_\ell'$, where $M'$ is the lower right
$\bigl((m-r)d\bigr)\times\bigl((m-r)d\bigr)$ sub-matrix of $M$, $(c^\ell)'$ are the remaining $(m-r)d$ entries
of $c^\ell$, and $\bar{a}_\ell'$ is a vector obtained from $\bar{a}_\ell$ by taking
into account the first $dr$ fixed values of $c^\ell$. All absolute values of
the entries of $\bar{a}_\ell'$ are bounded by $2r+1$.
As $M'$ is a diagonal block matrix, we can decompose this into $m-r$
systems with $d$ variables and equations each.
As every $d\times d$-block on the diagonal of the matrix $M'$ is an
orthonormal basis of the corresponding $d$-dimensional subspace, the
matrices in the subsystems are orthonormal. Furthermore, the right-hand sides
have a norm of at most $(2r+1)\sqrt{d}$. Hence, we can conclude that
$c_{\max}$ is bounded from above by $3\sqrt{d}r$.

Thus, the probability of the event $\calF$ can be bounded from above
by
\[
   \prod_{i=r+1}^m \frac{|\calI'_i|}{\sqrt{2\pi}\sigma}
   \le \left(\frac{6\sqrt{d}r(m-r)\lambda}{\sqrt{2\pi}\sigma}\right)^{m-r}
   \le \left(\frac{\sqrt{d}m^2\lambda}{\sigma}\right)^{m-r}\:,
\]
where we used that $r(m-r)\le m^2/4$.
Using Fact~\ref{fact:kmubProbExponents},
we can replace the exponent $m-r$ by a lower bound.
If all nodes are unbalanced, then $r$ equals the number of nodes minus one.
Otherwise, if $b\ge 1$, then $r$ equals the number of nodes minus $b$.
Hence, Lemma~\ref{NumberNodesBlueprint} yields
\[
  \Pr{\Lambda(\blueprint)\le \lambda} \le
  \begin{cases}
  \left(\frac{\sqrt{d}m^2\lambda}{\sigma}\right)^{\frac m6
   + \frac{\const_2-4}3+1} & \text{if $b=0$,}\\
   \left(\frac{\sqrt{d}m^2\lambda}{\sigma}\right)^{\frac m6
   +\frac{(2\const_1d-1)b-2}3+b} & \text{if $b\ge 1$,}
  \end{cases}
\]
which completes the proof.
\end{proof}

With the previous lemma, we can bound the probability that
there exists an iteration whose transition blueprint does not fall into any
of the previous categories and that makes a small improvement.

\begin{proof}[Proof of Lemma~\ref{lemma:Delta6}]
    Let $\mathbb{B}$ denote the set of $(m,b,\eps)$-blueprints that do not fall into
    the previous five categories. Here, $\eps$ is fixed
    but there are $nk$ possible choices for $m$ and $b$.
    As in the proof of Lemma \ref{lem:Delta3}, we will use a union bound to
    estimate the probability that there exists a blueprint $\blueprint \in \mathbb{B}$
    with $\Lambda(\blueprint) \le \lambda$. Note that once $m$ and $b$ are fixed, there are
    at most $(nk^2)^m$ possible choices for the edges in a blueprint, and for
    every balanced cluster, there are at most
    $\left(\frac{D\sqrt{d}}{\sqrt{n\eps}}\right)^d$ choices for its approximate center.
    Also, in all cases,
    $m \ge \max(\const_2 + 1, b(dz_1 + 1)) = \max(8, 8bd+b)$, because there
    is always one vertex with degree at least
    $\const_2 + 1$, and there are always $b$ vertices with degree at least $2dz_1 + 2$.

    Now we set $Y = k^5 \cdot \sqrt{ndD}$. Lemma~\ref{lemma:FixedBlueprint} yields
    the following bound:
    \begin{eqnarray}
        && \Prob\left[\exists \blueprint \in \mathbb{B} \Big| \Lambda(\blueprint) \le \frac{6D\sqrt{nd}}{Y} \cdot \eps^{1/3}\right] \notag\\
        &\le& \sum_{m=8}^n (nk^2)^m \cdot \left(\frac{6m^2dD\sqrt{n}}{Y\sigma} \cdot \eps^{1/3}\right)^{\frac{m}{6} + \frac{\const_2-1}{3}} \notag \\
        && + \sum_{b=1}^k\sum_{m=8bd+b}^n \left(\frac{D\sqrt{d}}{\sqrt{n\eps}}\right)^{bd} \cdot (nk^2)^m \cdot
           \left(\frac{6m^2dD\sqrt{n}}{Y\sigma} \cdot \eps^{1/3}\right)^{\frac{m}{6} + \frac{(2z_1d+2)b-2}{3}} \label{eqn:kmubCase5}.
    \end{eqnarray}

    Each term in the first sum simplifies as follows:
    \begin{eqnarray*}
        (nk^2)^m \cdot \left(\frac{6m^2dD\sqrt{n}}{Y\sigma} \cdot
        \eps^{1/3}\right)^{\frac{m}{6} + \frac{\const_2-1}{3}}
        &\le& \left(\frac{6n^{17/2}k^{12}dD}{Y\sigma} \cdot
        \eps^{1/3}\right)^{\frac{m}{6} + \frac{\const_2-1}{3}}\\
        &=& \left(\frac{6n^8k^7d^{1/2}D^{1/2}}{\sigma} \cdot \eps^{1/3}\right)^{\frac{m}{6} + \frac{\const_2-1}{3}}.
    \end{eqnarray*}
    Furthermore, $\frac{m}{6} + \frac{\const_2 - 1}{3} \ge \frac{8}{6} + \frac{6}{3} > 3$, so we can use Fact~\ref{fact:kmubProbExponents}
    to decrease the exponent here, which gives us
    \begin{eqnarray*}
        \left(\frac{6n^8k^7d^{1/2}D^{1/2}}{\sigma} \cdot \eps^{1/3}\right)^3
        &=& \eps \cdot \left(\frac{O(1) \cdot n^{24}k^{21}d^{3/2}D^{3/2}}{\sigma^3}\right).
    \end{eqnarray*}

    Similarly, each term in the second sum simplifies as follows:
    \begin{eqnarray*}
        && \left(\frac{D\sqrt{d}}{\sqrt{n\eps}}\right)^{bd} \cdot (nk^2)^m \cdot \left(\frac{6m^2dD\sqrt{n}}{Y\sigma} \cdot \eps^{1/3}\right)^{\frac{m}{6} + \frac{(2z_1d+2)b-2}{3}}\\
        &\le& \left(\frac{D\sqrt{d}}{\sqrt{n\eps}}\right)^{bd} \cdot \left(\frac{6n^8k^7d^{1/2}D^{1/2}}{\sigma} \cdot \eps^{1/3}\right)^{\frac{m}{6} + \frac{(2z_1d+2)b-2}{3}}.
    \end{eqnarray*}
    Furthermore,
    \begin{eqnarray*}
        \frac{m}{6} + \frac{(2z_1d+2)b-2}{3} \ge \frac{8bd+b}{6} + \frac{16bd+2b-2}{3} \ge \frac{20bd}{3}.
    \end{eqnarray*}
    Therefore, we can further bound this quantity by
    \begin{eqnarray*}
        && \left( \left(\frac{D\sqrt{d}}{\sqrt{n\eps}}\right)^{3/20} \cdot \frac{6n^8k^7d^{1/2}D^{1/2}}{\sigma} \cdot \eps^{1/3}\right)^{\frac{m}{6} + \frac{(2z_1d+2)b-2}{3}}\\
        &=& \left(\frac{6n^{317/40}k^{7}d^{23/40}D^{13/20}}{\sigma} \cdot \eps^{31/120}\right)^{\frac{m}{6} + \frac{(2z_1d+2)b-2}{3}}.
    \end{eqnarray*}
    As noted above,
    \begin{eqnarray*}
        \frac{m}{6} + \frac{(2z_1d+2)b-2}{3} \ge \frac{20bd}{3} > \frac{120}{31},
    \end{eqnarray*}
    so we can use Fact~\ref{fact:kmubProbExponents} to decrease the exponent, which gives us
    \begin{eqnarray*}
        \eps \cdot \left(\frac{6n^{317/40}k^7d^{23/40}D^{13/20}}{\sigma}\right)^{120/31}
        &<& \eps \cdot \left(\frac{O(1) \cdot n^{317/10}k^{28}d^{23/10}D^{13/5}}{\sigma^4}\right).
    \end{eqnarray*}

    Using these bounds, we can simplify equation \eqref{eqn:kmubCase5}:
    \begin{eqnarray*}
        && \Prob\left[\exists \blueprint \in \mathbb{B} \Big| \Lambda(\blueprint) \le \frac{6D\sqrt{nd}}{Y} \cdot \eps^{1/3}\right]\\
        &\le&
        \eps \cdot n \cdot \left(\frac{O(1) \cdot n^{24}k^{21}d^{3/2}D^{3/2}}{\sigma^3}\right) +
        \eps \cdot nk \cdot \left(\frac{O(1) \cdot n^{317/10}k^{28}d^{23/10}D^{13/5}}{\sigma^4}\right)\\
        &\le& \eps \cdot \left(\frac{O(1) \cdot n^{327/10}k^{29}d^{23/10}D^{13/5}}{\sigma^4}\right).
    \end{eqnarray*}
    On the other hand $Y = k^5 \cdot \sqrt{ndD} \ge 1$, so Lemma~\ref{lemma:kmubCloseCenters} guarantees
    \begin{eqnarray*}
        \Prob\left[\delta_\eps \le Y\eps^{1/6}\right] &\le& \eps \cdot \left(\frac{O(1) \cdot n^5Y}{\sigma}\right)^6\\
        &=& \eps \cdot \left(\frac{O(1) \cdot n^{11/2}k^5d^{1/2}D^{1/2}}{\sigma}\right)^6\\
        &=& \eps \cdot \left(\frac{O(1) \cdot n^{33}k^{30}d^3D^3}{\sigma^6}\right).
    \end{eqnarray*}

    Finally, we know from Lemma~\ref{lemma:kmubApproximateBisectors} that if a blueprint $\blueprint$ can result in a potential drop of at most
    $\eps$, then $\delta_\eps \cdot \Lambda(\blueprint) \le 6D\sqrt{nd\eps}$. We must therefore have either $\delta_\eps \le Y\eps^{1/6}$
    or $\Lambda(\blueprint) \le \frac{6D\sqrt{nd}}{Y} \cdot \eps^{1/3}$. Therefore,
    \begin{eqnarray*}
        \Pr{\Delta_6\le\eps}
        &\le& \Prob\left[\exists \blueprint \in \mathbb{B} \Big| \Lambda(\blueprint) \le \frac{6D\sqrt{nd}}{Y} \cdot \eps^{1/3}\right] +
        \Prob\left[\delta_\eps \le Y\eps^{1/6}\right] \\
        &\le& \eps \cdot \left(\frac{O(1) \cdot n^{327/10}k^{29}d^{23/10}D^{13/5}}{\sigma^4}\right) +
        \eps \cdot \left(\frac{O(1) \cdot n^{33}k^{30}d^3D^3}{\sigma^6}\right)\\
        &\le& \eps \cdot \left(\frac{O(1) \cdot n^{33}k^{30}d^3D^3}{\sigma^6}\right),
    \end{eqnarray*}
    which concludes the proof.
\end{proof}

\subsection{The Main Theorem}
\label{ssec:merging}

Given the analysis of the different types of iterations,
we can complete the proof that $k$-means has polynomial
smoothed running time.

\begin{proof}[Proof of Theorem~\ref{thm:main}]
    Let $T$ denote the maximum number of iterations that $k$-means can need
    on the perturbed data set $X$, and let $\Delta$ denote the minimum
    possible potential drop over a period of three consecutive iterations. As remarked in
    Section~\ref{sec:prelim},
    we can assume that all the data points lie in the hypercube $[-D/2, D/2]^d$ for
    $D = \sqrt{90kd \cdot \ln(n)}$, because
    the alternative contributes only an additive term of +1 to $\Ex T$.

    After the first iteration, we know $\pot \le ndD^2$. This implies that if
    $T \ge 3t+1$, then $\Delta \le ndD^2/t$.
    However, in the previous section, we proved that for $\eps \in (0,1]$,
    \begin{eqnarray*}
        \Prob[\Delta \le \eps]
        &\le& \sum_{i=1}^6\Pr{\Delta_i\le\eps} \\
        &\le& \eps \cdot \frac{O(1) \cdot n^{33}k^{30}d^3D^3}{\sigma^6}.
    \end{eqnarray*}
    Recall from Section~\ref{sec:prelim} that $T \le
    n^{3kd}$ regardless of the perturbation. Therefore,
    \begin{eqnarray*}
        \Ex T
        &\le& O(ndD^2) + \sum_{t=ndD^2}^{n^{3kd}} 3 \cdot P[T \ge 3t+1]\\
        &\le& O(ndD^2) + \sum_{t=ndD^2}^{n^{3kd}} 3 \cdot P\left[\Delta \le \frac{ndD^2}{t}\right]\\
        &\le& O(ndD^2) + \sum_{t=ndD^2}^{n^{3kd}} \frac{3ndD^2}{t} \cdot \left(\frac{O(1) \cdot n^{33}k^{30}d^3D^3}{\sigma^6}\right)\\
        &=& O(ndD^2) + \left(\frac{O(1) \cdot n^{34}k^{30}d^4D^5}{\sigma^6}\right) \cdot \left(\sum_{t=ndD^2}^{n^{3kd}} \frac{1}{t}\right)\\
        &=& O(ndD^2) + \left(\frac{O(1) \cdot n^{34}k^{30}d^4D^5}{\sigma^6}\right) \cdot O(kd \cdot \ln(n))\\
        &=& \frac{O(1) \cdot n^{34}k^{34}d^{8}\cdot \ln^4(n)}{\sigma^6},
    \end{eqnarray*}
which completes the proof.
\end{proof}

\section{Concluding Remarks}
\label{sec:conclusions}

In this paper, we settled the smoothed running time of the $k$-means
method for $d\ge 2$. For $d=1$, it was already known
that $k$-means has polynomial smoothed running
time~\cite{MantheyRoeglin:kMeans:2009}.

The exponents in our smoothed analysis are constant but large.
We did not make a huge effort to optimize the exponents as the
arguments are intricate enough even without trying to optimize constants.
Furthermore, we believe that our approach, which is essentially based
on bounding the smallest possible improvement in a single step,
is too pessimistic to yield a bound that matches experimental
observations. A similar phenomenon occurred already in the smoothed
analysis of the 2-opt heuristic for the TSP~\cite{EnglertRV07}.
There it was possible to improve the bound for the number of iterations
by analyzing sequences of consecutive steps rather than single steps.
It is an interesting question if this approach also leads to an improved
smoothed analysis of $k$-means.

Squared Euclidean distances, while most natural, are not the only distance
measure used for $k$-means clustering. The $k$-means
method can be generalized to arbitrary Bregman divergences~\cite{BanerjeeEA:Bregman:2005}.
Bregman divergences include the Kullback-Leibler divergence, which is used, e.g.,
in text classification, or Mahalanobis distances.
Due to its role in applications, $k$-means clustering
with Bregman divergences has attracted a lot of attention
recently~\cite{AckermannBloemer,AckermannEA}.
Since only little is known about the performance of the $k$-means method
for Bregman divergences, we raise the question how the $k$-means method
performs for Bregman divergences in the worst and smoothed case.

\bibliographystyle{plain}
\bibliography{Polynomial}

\end{document}